\documentclass[12pt,a4paper]{amsart}
\usepackage[latin1]{inputenc}
\usepackage{graphicx}
\usepackage{amssymb}
\usepackage{amsthm}
\usepackage{hyperref}
\usepackage[left=3cm,right=3cm,top=2.5cm,bottom=2.5cm]{geometry}
\usepackage{physics}
\usepackage{amsaddr}
\usepackage{enumerate}
\usepackage[T1]{fontenc}
\usepackage{xcolor}

\usepackage[abs]{overpic}
\usepackage[rightcaption]{sidecap}

\newtheorem{theorem}{Theorem}
\newtheorem{lemma}[theorem]{Lemma}

\newtheorem{proposition}[theorem]{Proposition}
\theoremstyle{definition}
\newtheorem{remark}[theorem]{Remark}

\newcommand{\Hil}{\mathcal{H}}
\newcommand{\K}{\mathbb{K}}
\newcommand{\R}{\mathbb{R}}
\newcommand{\C}{\mathbb{C}}
\newcommand{\T}{\mathbb{T}}
\newcommand{\Lin}{\mathcal{L}}
\newcommand{\Cha}{\mathcal{C}}

\newcommand{\ketbrar}[1]{\ketbra{#1}{#1}}

\newcommand{\mc}[1]{\mathcal{#1}}

\newcommand{\mb}[1]{\mathbb{#1}}

\newcommand{\N}{\mathbb N}

\definecolor{lighter}{RGB}{10.0,146.0,35.0}
\definecolor{darker}{RGB}{2.0,100.0,20.0}
\definecolor{sininen}{RGB}{89.5,101.8,238.1}
\definecolor{punainen}{RGB}{227.0,12.0,12.0}

\newcommand{\punainen}[1]{\textcolor{punainen}{#1}}

\begin{document}

    \title[Cond. for Large-Sample Majorization in Terms of $\alpha$-$z$ Rel. Entropies]{Conditions for Large-Sample Majorization of Pairs of Flat States in Terms of $\alpha$-$z$ Relative Entropies}

    \author{Frits Verhagen\textsuperscript{1}, Marco Tomamichel\textsuperscript{1,2}\\
        and Erkka Haapasalo\textsuperscript{1}}
	\address{1. Centre for Quantum  Technologies,  National University of Singapore\\
        2. Department of Electrical and Computer Engineering, National University of Singapore}
    
	\date{\today}
	
    \begin{abstract}
        We offer the first operational interpretation of the $\alpha$-$z$ relative entropies, a measure of distinguishability between two quantum states introduced by Jak\v{s}i\'{c} {\it et al.}~and Audenaert and Datta. We show that these relative entropies appear when formulating conditions for large-sample or catalytic majorization of pairs of flat states and certain generalizations of them. Indeed, we show that such transformations exist if and only if all the $\alpha$-$z$ relative entropies for $\alpha<1$ of the two pairs are ordered. In this setting, the $\alpha$ and $z$ parameters are truly independent from each other. These results also yield an expression for the optimal rate of converting one flat state pair into another. Our methods use real-algebraic techniques involving preordered semirings and certain monotone homomorphisms and derivations on them. 
    \end{abstract}

	\maketitle

\section{Introduction}


We study the following framework of \emph{majorization} between two quantum state pairs: Consider two systems with Hilbert spaces $\Hil_{\rm in}$ and $\Hil_{\rm out}$ and pairs of states (density operators) $(\rho,\sigma)$ on $\Hil_{\rm in}$ and $(\rho',\sigma')$ on $\Hil_{\rm out}$. We want to know if there is a single quantum channel (a completely positive trace-preserving linear map) $\mc C$ such that
\begin{equation}\label{eq:maj}
\mc C(\rho)=\rho',\qquad \mc C(\sigma)=\sigma'.
\end{equation}
This setting is encountered, e.g.,\ in quantum thermodynamics where we ask if a state $\rho$ can be transformed into $\rho'$ using a Gibbs-preserving map. In this case, both $\sigma$ and $\sigma'$ are equal to $\gamma_\beta$, the Gibbs state with inverse temperature $\beta$. 

Instead of the single-shot case of \eqref{eq:maj}, we are mainly concerned with {\it large-sample} and {\it catalytic} majorization. For the former, we ask if, for all $n\in\N$ sufficiently large, there exists a channel $\mc C_n$ such that
\begin{equation}\label{eq:LS}
\mc C_n\big(\rho^{\otimes n}\big)=(\rho')^{\otimes n},\qquad \mc C_n\big(\sigma^{\otimes n}\big)=(\sigma')^{\otimes n}.
\end{equation}
For catalytic majorization, we ask if there exists a pair $(\tau,\omega)$ of states (catalysts) and a channel $\mc C$ such that
\begin{equation}\label{eq:catalytic}
\mc C(\rho\otimes\tau)=\rho'\otimes\tau,\qquad \mc C(\sigma\otimes\omega)=\sigma'\otimes\omega.
\end{equation}
If the catalyzing states $\tau$ and $\omega$ were mutually orthogonal, they could catalyze any transformations. This is why we always assume that a pair $(\tau,\omega)$ of catalysts are not completely orthogonal (i.e.,\ their fidelity $F(\tau,\omega)>0$). Large-sample majorization implies catalytic majorization \cite{Duan2005}, but the converse is not typically true \cite{Feng_et_al_2006}. In addition to the exact scenario above, we will also be interested in \emph{asymptotic} large-sample and catalytic majorization. In this case, we require that the target output states are only approximately reached by a channel, but up to arbitrary precision. 

The aim of this work is to identify conditions in the form of inequalities
\begin{equation}\label{eq:suff}
\mb D(\rho\|\sigma)>\mb D(\rho'\|\sigma')\text{ or }\mb D(\rho\|\sigma)\geq\mb D(\rho'\|\sigma'), 
\end{equation}
for certain real-valued maps $\mb D$ on pairs of states, guaranteeing either exact or asymptotic, large-sample or catalytic, majorization. It will turn out that the strict inequality above is applicable to the exact scenario, and the non-strict inequality to the asymptotic one. The quantities $\mb D$ should naturally be monotone under majorization, i.e.,\ if $(\rho,\sigma)$ majorizes $(\rho',\sigma')$ in the sense of \eqref{eq:maj}, then $\mb D(\rho\|\sigma)\geq\mb D(\rho'\|\sigma')$. In other words, $\mb D$ should satisfy the {\it data-processing inequality} (DPI). It is also natural to require that these quantities are {\it (tensor) additive} or {\it extensive} in the sense that
\begin{equation}
\mb D(\rho_1\otimes\rho_2\|\sigma_1\otimes\sigma_2)=\mb D(\rho_1\|\sigma_1)+\mb D(\rho_2\|\sigma_2).
\end{equation}
Together with DPI, additivity implies that when one pair of states majorizes another pair in the large-sample or catalytic setting, then the non-strict inequalities in \eqref{eq:suff} will be met. Quantities like $\mb D$ are often called {\it relative entropies}. Many quantum relative entropies have been presented in the literature. However, while all \emph{classical} relative entropies have been identified \cite{Mu_et_al_2020}, finding all possible quantum relative entropies is an open problem. Indeed, the non-commutativity of quantum states results in many quantum generalizations of a single classical relative entropy, and finding a complete characterization of all monotone generalizations is currently still intractable. 

Our partial solution to this problem is restricting the kinds of pairs of quantum states that we study. Recall that a state $\rho$ is a \emph{classical-quantum (cq) state} if there are probabilities $p_i$ and states $\rho_i$ such that
\begin{equation}
\rho=\sum_{i=1}^n p_i \ketbrar{i}\otimes\rho_i
\end{equation}
with fixed orthonormal flags $\ket{i}$, and $n\geq1$. We will specialize to pairs of such states where additionally all component states $\rho_i$ are pure, i.e. pairs $(\rho,\sigma)$ such that 
\begin{equation}\label{eq:pairform}
\rho=\sum_{i=1}^n p_i \ketbrar{i}\otimes\ketbrar{\alpha_i},\quad\sigma=\sum_{i=1}^n q_i \ketbrar{i}\otimes\ketbrar{\beta_i}. 
\end{equation}
Here, $\ket{\alpha_i}$ is a normalized vector when $p_i>0$ and the zero vector when $p_i=0$, for $i=1,\ldots,n$, and similarly for $\ket{\beta_i}$, $q_i$. We say that $\rho$ and $\sigma$ of \eqref{eq:pairform} {\it have some overlap} if there is $i$ such that $\langle\alpha_i|\beta_i\rangle\neq 0$. We also say that $\rho$ and $\sigma$ are {\it non-parallel} if $|\braket{\alpha_i}{\beta_i}|<1$ for all $i$. Let us point out a particular subset of pairs $(\rho,\sigma)$ like those in \eqref{eq:pairform}: Assume that $\rho=\tr[P]^{-1}P$ and $\sigma=\tr[Q]^{-1}Q$ where $P$ and $Q$ are projections. Due to Jordan's Lemma, this pair $(\rho,\sigma)$ has the desired cq-decompositions with pure component states. Naturally, weighted direct sums of these {\it flat states} can also be decomposed in the same manner. For technical reasons, we will exclusively consider pairs of cq-states with pure components as in \eqref{eq:pairform} that have some overlap, and denote this collection of pairs by $\mc F$. Note that any $(\rho,\sigma)\in\mc F$ as in \eqref{eq:pairform} is essentially finite dimensional since $n<\infty$ and the vectors $\ket{\alpha_i}$ and $\ket{\beta_i}$ span a 2-dimensional space for each $i$.

In the above restricted setting, we identify a set of relative entropies $\mb D$ giving sufficient and (almost) necessary conditions as in \eqref{eq:suff} for exact or asymptotic, large-sample or catalytic majorization. These relative entropies turn out to be within the set of the {\it $\alpha$-$z$ relative entropies} $D_{\alpha,z}$. The $\alpha$-$z$ relative entropies for any pair of quantum states $(\rho,\sigma)$, introduced in \cite{Jaksic2012} and \cite{Audenaert_Datta_2015} (see also a new generalization in \cite{Hiai_Jencova_2024}), are defined as 
\begin{equation}\label{eq:alphazdivergence}
    D_{\alpha,z}(\rho\|\sigma):=\frac{1}{\alpha-1}\log\tr\left[\left(\sigma^\frac{1-\alpha}{2z}\rho^\frac{\alpha}{z}\sigma^\frac{1-\alpha}{2z}\right)^z\right], 
\end{equation}
as well as certain limits $\alpha\rightarrow1$, $z\rightarrow0$, of this expression, and with some restrictions on the supports of $\rho$, $\sigma$. The question of what the precise range of $(\alpha,z)$ is where the $D_{\alpha,z}$ satisfy DPI, has been investigated in a number of papers \cite{Audenaert_Datta_2015,Hiai2013,Carlen_et_al_2018} and conclusively answered in \cite{Zhang2020}. 

The relative entropies $\mathbb{D}$ giving conditions of the form \eqref{eq:suff} for exact or asymptotic, large-sample or catalytic majorization of flat states (and their generalizations) are essentially those $D_{\alpha,z}$ with parameters $\alpha$ and $z$ such that 
\begin{equation}\label{eq:parameterset}
    \alpha\in(0,1)\text{ and }z\geq\max\{\alpha,1-\alpha\}, 
\end{equation}
together with certain extensions to $\alpha=0,1$ and $z\rightarrow\infty$. Note that the $D_{\alpha,z}$ with $\alpha>1$ do not play a role here because they would diverge on the pairs of states we consider, since the component states are pure. To our knowledge, this is the first time an operational interpretation of the $\alpha$-$z$ relative entropies, where $\alpha$ and $z$ are truly independent of each other, has been demonstrated. Let us note however that the special cases of the so-called `sandwiched' quantum relative entropies \cite{Muller-Lennert_et_al_2013,Wilde2013StrongCF}, where $\alpha=z$, and the Petz-type relative entropies \cite{Petz_85,Petz_86}, where $z=1$, have their own fields of applications and operational interpretations.

For the setting of asymptotic catalytic majorization between pairs of states in $\mc F$, we show the following result, which gives necessary and sufficient conditions in terms of $\alpha$-$z$ relative entropies. 

\begin{theorem}\label{thm:approximateC}
    Let $(\rho,\sigma)$, $(\rho',\sigma')$ be two pairs of states in $\mc F$ (e.g. flat states), where $\rho$ and $\sigma$ are non-parallel and $\rho'$, $\sigma'$ do not commute. Then the following are equivalent: 
    \begin{enumerate}[(i)]
        \item $D_{\alpha,z}(\rho\|\sigma)\geq D_{\alpha,z}(\rho'\|\sigma')$ when $\alpha\in(0,1)$ and $z>\max\{\alpha,1-\alpha\}$. 
        \item For each $\varepsilon>0$, there exist a state $\rho'_\varepsilon$ such that $F(\rho'_\varepsilon,\rho')\geq1-\varepsilon$, a quantum channel $\mc C_\varepsilon$, and a catalyzer pair $(\tau_\varepsilon,\omega_\varepsilon)\in\mc F$ such that 
        \begin{equation}
            \mc C_\varepsilon(\rho\otimes\tau_\varepsilon)=\rho'_\varepsilon\otimes\tau_\varepsilon\quad\text{and}\quad \mc C_\varepsilon(\sigma\otimes\omega_\varepsilon)=\sigma'\otimes\omega_\varepsilon. 
        \end{equation}
    \end{enumerate}
\end{theorem}

Additionally, we show that the range $\alpha\in(0,1)$ and $z>\max\{\alpha,1-\alpha\}$ for the $\alpha$-$z$ relative entropies appearing in (i) above is minimal, in the sense that if the range in (i) is made smaller by removing any non-empty open subset, it is not a sufficient condition for (ii) to hold. See Theorem \ref{thm:minimal} in the next section. 

Theorem \ref{thm:approximateC} has a counterpart (Theorem \ref{thm:approximateLS}), dealing with the asymptotic large-sample case, where we need to introduce an extra condition on $(\rho,\sigma)$. Additionally, we show similar results giving sufficient conditions for \emph{exact} large-sample and catalytic majorization (Theorem \ref{thm:LS}), for which we will also need the earlier mentioned extensions of $D_{\alpha,z}$ to $\alpha=0,1$ and $z\rightarrow\infty$. These conditions are almost always necessary as well, in a sense that will be made more precise later. 

Our results also yield the optimal rate of converting a pair $(\rho,\sigma)$ of cq-states with pure components to another such pair $(\rho',\sigma')$: it is given by the infimum of the ratio $D_{\alpha,z}(\rho\|\sigma)/D_{\alpha,z}(\rho'\|\sigma')$ over all $\alpha\in(0,1)$, $z>\max\{\alpha,1-\alpha\}$. See Theorem \ref{thm:rates}. 


The framework of majorization we study in this work generalizes the classical notion of the comparison of statistical experiments defined as in \eqref{eq:maj} where $\rho$, $\rho'$, $\sigma$, and $\sigma'$ commute (i.e.\ are probability distributions) and $\mc C$ is a stochastic map. This classical majorization can be equivalently characterized by the majorization of Lorenz curves \cite{Blackwell53}. The latter notion was generalized to the quantum case using quantum relative Lorenz curves in \cite{Buscemi2017}, where it was called \emph{quantum relative majorization}. Our notion of majorization characterized by \eqref{eq:maj} implies quantum relative majorization. In the case of pairs of qubit states the two notions coincide \cite{Alberti1980}, but in general, they are not equivalent, unlike in the classical case \cite{matsumoto2014}. 

The techniques we use in proving our results come from the real-algebraic theory of preordered semirings recently developed by Tobias Fritz in \cite{fritz2022,fritz2023}. The theorems presented there, collectively known as the \emph{Vergleichsstellens\"atze}, allow us to identify the relative entropies involved in the conditions for majorization. These methods can be seen as an extension of the theory of asymptotic spectra previously developed by Volker Strassen, who used his techniques to introduce the first subcubic algorithm for matrix multiplication \cite{strassen1986,strassen1987,strassen1988,strassen1991}. Subsequently, Strassen's theory has found applications in probability and information theory, and computer science. Another extension of this theory was presented in \cite{Vrana2022}. Recently, these techniques have been successfully applied to settings involving quantum states \cite{bunth2021asymptotic,perry2022semiring,bunth2023,bunth2024}. Our current results also build on our previous work in \cite{farooq2024,verhagen2025}. 

This paper is organized as follows: In Section \ref{sec:results}, we state the remaining main results that were outlined above. To prepare for the proof of our results, we introduce some notation and necessary mathematical tools in Section \ref{sec:bg}. In Section \ref{sec:technical} we apply the methodology introduced in Section \ref{sec:bg} to our specific setting. We use the technical results of Section \ref{sec:bg} in Section \ref{sec:proofs} to derive our main results on large-sample or catalytic, and exact or asymptotic, majorization, and the optimal conversion rate.

\section{The Main Results}\label{sec:results}

    We now present our remaining results. For a pair of states $(\rho,\sigma)\in\mc F$, and $\alpha\in[0,1]$, $z\geq\max\{\alpha,1-\alpha\}$, define     
    \begin{equation}\label{eq:alphazexplicit}
        \hat{D}_{\alpha,z}(\rho\|\sigma):=-\frac{1}{z+1}\log\left(\sum_{i=1}^n(p_i)^\alpha(q_i)^{1-\alpha}|\braket{\alpha_i}{\beta_i}|^{2z}\right), 
    \end{equation}
    and 
    \begin{equation}\label{eq:DT}
        \hat{D}^\T(\rho\|\sigma):=-\log\left(\max_{i=1,\ldots,n}|\braket{\alpha_i}{\beta_i}|^{2}\right). 
    \end{equation}
    In \eqref{eq:alphazexplicit} for $\alpha=0$ or $1$, we adopt the convention that $0^0:=0$. When $\alpha\in(0,1)$, the $\hat{D}_{\alpha,z}$ coincide with the $\alpha$-$z$ relative entropies $D_{\alpha,z}$ as in \eqref{eq:alphazdivergence}, when evaluated on pairs of states in $\mc F$, except for a different prefactor in front of the $\log$ for technical reasons. Since the $D_{\alpha,z}$ satisfy DPI for $\alpha\in(0,1)$, $z\geq\max\{\alpha,1-\alpha\}$, the same holds for $\hat{D}_{\alpha,z}$, where  we use the continuity of $\hat{D}_{\alpha,z}$ in $\alpha$ to establish the DPI when $\alpha=0,1$. $\hat{D}^\T$ is the point-wise limit $z\rightarrow\infty$ of $\hat{D}_{\alpha,z}$, for any $\alpha\in[0,1]$ as one may rather easily verify. As a consequence, also $\hat{D}^\T$ satisfies DPI. See Figure \ref{fig:alphazregion} for a depiction of the family of all these relative entropies. 

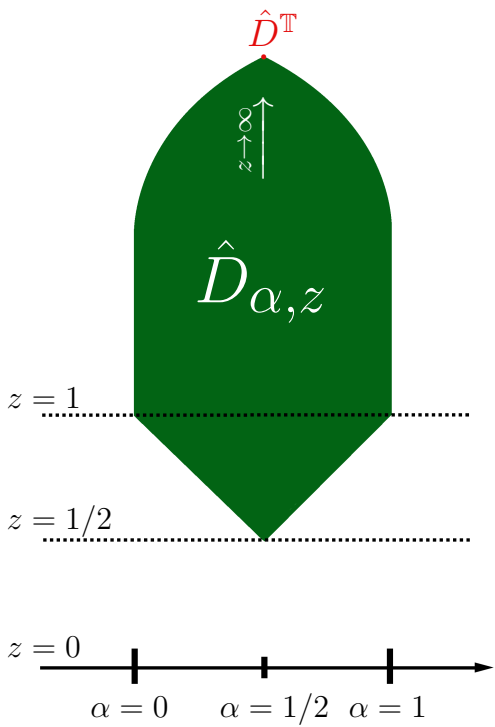
\begin{SCfigure}[50][b]
\begin{overpic}[scale=0.4,unit=1mm]{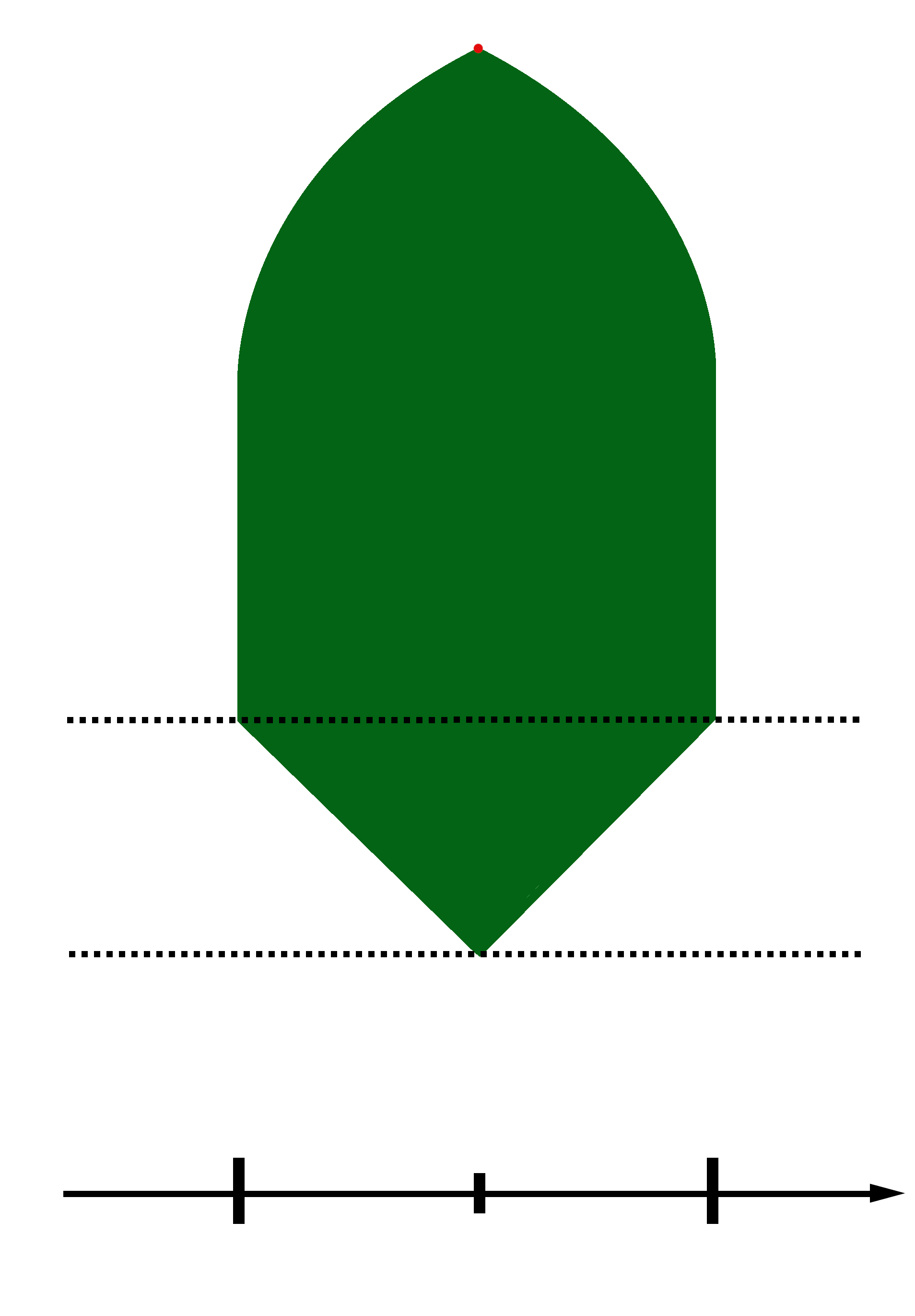}
\put(0,43){$z=1$}
\put(0,27){$z=1/2$}
\put(0,10){$z=0$}
\put(11,2){$\alpha=0$}
\put(28,2){$\alpha=1/2$}
\put(45,2){$\alpha=1$}
\put(22,57){\begin{Huge}
$\textcolor{white}{\hat{D}_{\alpha,z}}$
\end{Huge}}
\put(28.5,73){\begin{large}
\rotatebox{90}{$\textcolor{white}{\xrightarrow{z\rightarrow\infty}}$}
\end{large}}
\put(30,91){\begin{large}
\punainen{$\hat{D}^\T$}
\end{large}}
\end{overpic}
\caption{\label{fig:alphazregion} A depiction of the family of all relative entropies given by \eqref{eq:alphazexplicit} and \eqref{eq:DT}. For $\alpha\in(0,1)$ and $z\geq\max\{\alpha,1-\alpha\}$, $\hat{D}_{\alpha,z}$ equals the $\alpha$-$z$ relative entropies given by \eqref{eq:alphazdivergence}, up to a prefactor. With this different prefactor, the $\alpha$-$z$ divergences can be extended easily to the segments $\alpha=0$ or $\alpha=1$, and $z\geq1$. The additional relative entropy $\hat{D}^\T=\lim_{z\rightarrow\infty}\hat{D}_{\alpha,z}$, independent of $\alpha$, turns the set of all relative entropies into a compact space in the topology of point-wise convergence.}
\end{SCfigure}
    
    A result by Uhlmann \cite{Uhlmann1985} states that a pair of pure states $(\ketbrar{\alpha},\ketbrar{\beta})$ majorizes another pair of pure states $(\ketbrar{\alpha'},\ketbrar{\beta'})$ if and only if $|\braket{\alpha}{\beta}|\leq|\braket{\alpha'}{\beta'}|$. Note that in the expressions \eqref{eq:alphazexplicit} and \eqref{eq:DT}, the fidelity between pure states $|\braket{\alpha_i}{\beta_i}|$ also makes an appearance. 
    
    Recall that a pair $(\rho,\sigma)\in\mc F$ is called non-parallel if 
    \begin{equation}\label{eq:pu1}
        |\braket{\alpha_i}{\beta_i}|<1\quad\forall\;i\in\{1,\ldots,n\}. 
    \end{equation}
    To state our remaining results, we introduce the following additional property for a pair $(\rho,\sigma)\in\mc F$: 
    \begin{equation}\label{eq:pu2}
        \exists\;i\in\{1,\ldots,n\}\ \text{s.t.}\ \braket{\alpha_i}{\beta_i}=0\ \text{and}\ \ket{\alpha_i},\ket{\beta_i}\neq0.  
    \end{equation}

    The following Theorem gives sufficient and almost necessary conditions for exact large-sample and catalytic majorization. This result is established in Section \ref{sec:exactproof}.  
    
    \begin{theorem}\label{thm:LS}
        Let $(\rho,\sigma)$, $(\rho',\sigma')$ be two pairs of states in $\mc F$, where additionally $(\rho,\sigma)$ satisfies \eqref{eq:pu2}. If 
        \begin{align}
            \label{eq:LScondition1}&\hat{D}_{\alpha,z}(\rho\|\sigma)>\hat{D}_{\alpha,z}(\rho'\|\sigma')\ \text{when}\ \alpha\in[0,1],\ z\geq\max\{\alpha,1-\alpha\}, &\\
            \label{eq:LScondition2}&\hat{D}^\mathbb{T}(\rho\|\sigma)>\hat{D}^\mathbb{T}(\rho'\|\sigma'), 
        \end{align}            
        then $(\rho,\sigma)$ majorizes $(\rho',\sigma')$ both in the large-sample setting and in the catalytic setting, with a pair of catalyst states in $\mc F$. For the catalytic result, $(\rho,\sigma)$ does not need to satisfy \eqref{eq:pu2}. 
        
        Conversely, if $(\rho,\sigma)$ majorizes $(\rho',\sigma')$ in large samples or catalytically, then the above inequalities \eqref{eq:LScondition1} and \eqref{eq:LScondition2} hold non-strictly. 
    \end{theorem}

    Theorem \ref{thm:approximateC} on asymptotic catalytic majorization presented in the Introduction can be extended to the following result that also gives conditions for asymptotic large-sample majorization. The proof can be found in Section \ref{sec:approximate}. 

    \begin{theorem}\label{thm:approximateLS}
        Let $(\rho,\sigma)$, $(\rho',\sigma')$ be two pairs of states in $\mc F$, where $\rho$, $\sigma$ are non-parallel (i.e. satisfy \eqref{eq:pu1}) and $\rho'$, $\sigma'$ do not commute. Consider the following statements: 
        \begin{enumerate}[(i)]
            \item $D_{\alpha,z}(\rho\|\sigma)\geq D_{\alpha,z}(\rho'\|\sigma')$ when $\alpha\in(0,1)$ and $z>\max\{\alpha,1-\alpha\}$. 
            \item For each $\varepsilon>0$, there exist a state $\rho'_\varepsilon$ such that $F(\rho'_\varepsilon,\rho')\geq1-\varepsilon$, and a quantum channel $\mc{C}_{\varepsilon,n}$ for $n\geq n_\varepsilon$ large enough such that $\mc{C}_{\varepsilon,n}(\rho^{\otimes n})=(\rho'_\varepsilon)^{\otimes n}$ and $\mc{C}_{\varepsilon,n}(\sigma^{\otimes n})=(\sigma')^{\otimes n}$. 
            \item For each $\varepsilon>0$, there exist a state $\rho'_\varepsilon$ such that $F(\rho'_\varepsilon,\rho')\geq1-\varepsilon$, a quantum channel $\mc C_\varepsilon$, and a catalyzer pair $(\tau_\varepsilon$, $\omega_\varepsilon)\in\mc F$ such that 
            \begin{equation}
                \mc C_\varepsilon(\rho\otimes\tau_\varepsilon)=\rho'_\varepsilon\otimes\tau_\varepsilon\quad\text{and}\quad \mc C_\varepsilon(\sigma\otimes\omega_\varepsilon)=\sigma'\otimes\omega_\varepsilon. 
            \end{equation}
            Statements (i) and (iii) are equivalent. If $(\rho,\sigma)$ furthermore satisfies \eqref{eq:pu2}, then (i), (ii) and (iii) are equivalent. 
        \end{enumerate}        
    \end{theorem}

    Note that the state $\rho'_\varepsilon$ in statements (ii) and (iii) above is allowed to come from a Hilbert space that is larger than the one associated to $\rho'$. Additionally, $(\rho'_\varepsilon,\sigma')$ need not be a pair in $\mc F$. However, from the construction of $\rho'_\varepsilon$ in the proof of the implications (i) $\Longrightarrow$ (ii) and (i) $\Longrightarrow$ (iii), it follows that $\rho'_\varepsilon$ can be chosen such that $(\rho'_\varepsilon,\sigma')\in\mc F$, and $\rho'_\varepsilon$ and $\rho'$ have the same support. 
    
    It might be possible to relax the assumption \eqref{eq:pu2} in Theorems \ref{thm:LS} and \ref{thm:approximateLS}, but this requires the study of a slightly different semiring. We will address this issue later in the Discussion. Note that, for the results on large-sample majorization, Theorem \ref{thm:approximateLS} requires $(\rho,\sigma)$ to satisfy both \eqref{eq:pu1} and \eqref{eq:pu2}, while Theorem \ref{thm:LS} only explicitly requires \eqref{eq:pu2}. However, \eqref{eq:pu1} is automatically satisfied for $(\rho,\sigma)$ in Theorem \ref{thm:LS} since the strict inequality \eqref{eq:LScondition2} implies \eqref{eq:pu1}. 

    One may wonder whether the set of relative entropies that give sufficient conditions for large-sample and catalytic majorization in the preceding Theorem is bigger than necessary. The following Theorem tells us that this is in fact not the case, since the set of relative entropies is minimal in the sense that we cannot remove any open subset from the range of parameters 
    \begin{equation}\label{eq:range}
        R:=\big\{\,(\alpha,z)\in\R^2\,|\,\alpha\in(0,1),\ z>\max\{\alpha,1-\alpha\}\,\big\}
    \end{equation}
    appearing in condition (i) of Theorem \ref{thm:approximateLS}. The proof is given in Section \ref{sec:minimal}. 
    \begin{theorem}\label{thm:minimal}
         Let $O$ be a non-empty open subset of $R$. There exist pairs $(\rho,\sigma)$, $(\rho',\sigma')$ in $\mc F$, where $\rho$, $\sigma$ are non-parallel and satisfy \eqref{eq:pu2}, and $\rho'$, $\sigma'$ do not commute, such that 
         \begin{equation}\label{eq:smallercondition}
             D_{\alpha,z}(\rho\|\sigma)\geq D_{\alpha,z}(\rho'\|\sigma')\ \text{when}\ (\alpha,z)\in R\setminus O, 
         \end{equation}
         but both (ii) and (iii) of Theorem \ref{thm:approximateLS} do not hold. 
    \end{theorem}

    An analogous result can be shown for the case of exact large-sample and catalytic majorization (Theorem \ref{thm:LS}): one cannot remove any open subset from the set of all relative entropies appearing in the conditions \eqref{eq:LScondition1} and \eqref{eq:LScondition2}. 

    We say that, given pairs $(\rho,\sigma)$ and $(\rho',\sigma')$, $r\geq0$ is an {\it achievable conversion rate} if $r\leq\liminf_{n\to\infty}m_n/n$ where $(m_n)_n$ is a sequence of natural numbers such that $\big(\rho^{\otimes n},\sigma^{\otimes n}\big)$ majorizes $\big((\rho')^{\otimes m_n},(\sigma')^{\otimes m_n}\big)$ for large enough $n\in\N$. The {\it optimal conversion rate} $r\big((\rho,\sigma)\to(\rho',\sigma')\big)$, namely the supremum of all the achievable conversion rates, is given by the following result, which is shown in Section \ref{sec:rate}. 

    \begin{theorem}\label{thm:rates}
        For pairs $(\rho,\sigma)$ and $(\rho',\sigma')$ of states in $\mc F$, where $\rho$, $\sigma$ are non-parallel and satisfy \eqref{eq:pu2}, the optimal conversion rate from $(\rho, \sigma)$ to $(\rho', \sigma')$ is 
        \begin{align}
        &r\big((\rho,\sigma)\to(\rho',\sigma')\big)=\inf\left\{\,\frac{D_{\alpha,z}(\rho\|\sigma)}{D_{\alpha,z}(\rho'\|\sigma')}\,\middle|\,\alpha\in(0,1),z>\max\{\alpha,1-\alpha\}\,\right\}&\label{eq:rate1}\\
        &\quad=\min\left\{\,\frac{\hat{D}_{\alpha,z}(\rho\|\sigma)}{\hat{D}_{\alpha,z}(\rho'\|\sigma')}\,\middle|\,\alpha\in[0,1],z\geq\max\{\alpha,1-\alpha\}\,\right\}\cup\left\{\frac{\hat{D}^\T(\rho\|\sigma)}{\hat{D}^\T(\rho'\|\sigma')}\right\}.&\label{eq:rate2} 
        \end{align}
    \end{theorem}

    The additional relative entropies $\hat{D}_{\alpha,z}$ for $\alpha=0,1$ and the limit point $\hat{D}^\T$ make the set of relative entropies appearing in \eqref{eq:rate2} compact in the topology of pointwise convergence \cite[Proposition 8.5]{fritz2022}. This is why in \eqref{eq:rate2} we are able to take the minimum, instead of taking the infimum when using the original definition of the $\alpha$-$z$ relative entropies $D_{\alpha,z}$ in \eqref{eq:rate1}. 


    \section{Mathematical Background}\label{sec:bg}

    \subsection{The Vergleichsstellens\"atze}\label{sec:vss}

    To derive our results, we need some mathematical machinery that comes in a real-algebraic form and involves preordered semirings in particular. The minimal background required to use these techniques will be presented in this section. All definitions and results presented here are from \cite{fritz2022,fritz2023}, where one can find further details. 

Throughout this work, we follow the convention $\N=\{1,2,\ldots\}$ and denote \sloppy{${[n]:=\{1,\ldots,n\}}$} for all $n\in\N$. We will only study \emph{commutative} algebraic structures. A tuple $(S,+,\cdot,0,1,\preceq)$ is called a {\it preordered semiring} if $(S,+,0)$ and $(S,1,\cdot)$ are commutative semigroups where the multiplication distributes over the addition and $\preceq$ is a preorder (a binary relation which is reflexive and transitive) satisfying
\begin{equation}
x\preceq y\ \Rightarrow\ \left\{\begin{array}{rcl}
x+a&\preceq&y+a,\\
xa&\preceq&ya,
\end{array}\right.
\end{equation}
for all $a\in S$. As above, we usually omit the multiplication dot between elements in the semiring. We denote by $\sim$ the equivalence relation generated by $\preceq$, i.e. $x\sim y$ if and only if there are $z_1,\ldots,z_n\in S$ such that
\begin{equation}
x\preceq z_1\succeq z_2\preceq\cdots\succeq z_n\preceq y.
\end{equation}
The preordered semiring $S$ is of {\it polynomial growth} if it has a {\it power universal} $u\in S$, i.e.
\begin{equation}\label{eq:puproperty}
x\preceq y\quad\Rightarrow\quad\exists k\in\N:\ y\preceq xu^k.
\end{equation}
A preordered semiring $S$ is a {\it preordered semidomain} if
\begin{equation}
\begin{array}{rcl}
xy=0&\Rightarrow&x=0\ {\rm or}\ y=0,\\
0\preceq x\preceq 0&\Rightarrow&x=0.
\end{array}
\end{equation}
Moreover, $S$ is {\it zerosumfree} if $x+y=0$ implies $x=0=y$.
	
Given preordered semirings $S$ and $T$, we say that a map $\Phi:S\to T$ is a {\it monotone homomorphism} if
\begin{itemize}
\item $\Phi(x+y)=\Phi(x)+\Phi(y)$ for all $x,y\in S$ (additivity),
\item $\Phi(xy)=\Phi(x)\Phi(y)$ for all $x,y\in S$ (multiplicativity),
\item $x\preceq y$ $\Rightarrow$ $\Phi(x)\preceq\Phi(y)$ (monotonicity), and
\item $\Phi(0)=0$ and $\Phi(1)=1$.
\end{itemize}
We say that such a monotone homomorphism is {\it degenerate} if $x\preceq y$ implies $\Phi(x)=\Phi(y)$. Otherwise $\Phi$ is {\it non-degenerate}. For our results we need monotone homomorphisms with values in special semirings. These are the following:
\begin{itemize}
\item $\R_+$: The half-line $[0,\infty)$ equipped with the natural addition, multiplication, and total order.
\item $\R_+^{\rm op}$: The same as above but with the reversed order. Together $\R_+$ and $\R_+^{\rm op}$ are often called {\it temperate reals}.
\item $\T\R_+$: The half-line $[0,\infty)$ equipped with the natural multiplication, total order, and the tropical sum $x+y=\max\{x,y\}$.
\item $\T\R_+^{\rm op}$: The same as above but with the reversed order. Together $\T\R_+$ and $\T\R_+^{\rm op}$ are often called {\it tropical reals}.
\end{itemize}
Suppose that $S$ is a preordered semiring and that $\Phi:S\to\R_+$ is a monotone homomorphism. We say that an additive map $\Delta:S\to\R$ is a {\it derivation at $\Phi$} or a {\it $\Phi$-derivation} if it satisfies the {\it Leibniz rule}
\begin{equation}
\Delta(xy)=\Delta(x)\Phi(y)+\Phi(x)\Delta(y)
\end{equation}
for all $x,y\in S$. We are mainly only interested in derivations at degenerate homomorphisms that are also monotone, i.e. satisfy 
\begin{equation}
    x\preceq y\quad\Longrightarrow\quad\Delta(x)\leq\Delta(y). 
\end{equation}


Several results collectively called the ``Vergleichsstellens\"{a}tze'' have been derived in \cite{fritz2022}. Of all of them, we will need the following version:
	
\begin{theorem}[Based on Theorem 8.6 in \cite{fritz2022}]\label{thm:Fritz2022}
Let $S$ be a zerosumfree preordered semidomain with a power universal element $u$. Assume that for some $d\in\N$ there is a surjective homomorphism $\|\cdot\|:S\to\R_{>0}^d\cup\{(0,\ldots,0)\}$\footnote{Note that $\R_{>0}^d\cup\{(0,\ldots,0)\}$ with component-wise addition and multiplication forms a semiring.} with trivial kernel and such that
\begin{equation}\label{eq:surjectivehomomorphismproperties}
a\succeq b\ \Rightarrow\ \|a\|=\|b\|\quad {\rm and} \quad \|a\|=\|b\|\ \Rightarrow\ a\sim b.
\end{equation}
Denote the component homomorphisms of $\|\cdot\|$ by $\|\cdot\|_{(j)}$, $j=1,\ldots,d$. Let $x,y\in S\setminus\{0\}$ with $\|x\|=\|y\|$. If 
\begin{enumerate}[(i)]
\item for every $\mb K\in\{\R_+,\R_+^{\rm op},\T\R_+,\T\R_+^{\rm op}\}$ and every non-degenerate monotone homomorphism $\Phi : S \to \mb K$ with trivial kernel, we have $\Phi(x) > \Phi(y)$, and
\item $\Delta(x) > \Delta(y)$ for every monotone $\|\cdot\|_{(j)}$-derivation\footnote{\label{ftn:interchangeability}In fact, this inequality only has to be verified for every derivation up to interchangeability: two derivations of some degenerate monotone homomorphism are called \emph{interchangeable} if their difference takes a (possibly different) constant value on every equivalence class of the relation $\sim$.} $\Delta : S \to \R$ with $\Delta(u) = 1$ for all component indices $j = 1,\ldots,d$,
\end{enumerate}
then
\begin{enumerate}
\item[(a)] there exists a nonzero $a\in S$ such that $ax\succeq ay$, and 
\item[(b)] if additionally $x$ is power universal, then $x^n \succeq y^n$ for all sufficiently large $n\in\N$. 
\end{enumerate}
Conversely if either of these properties holds for any $n$ or $a$, then the above inequalities hold non-strictly.
\end{theorem}

Large-sample ordering as in (b) in the above Theorem implies catalytic ordering as in (a), with a catalyst of the form $a=\sum_{\ell=0}^{n-1} x^\ell y^{n-1-\ell}$ for a sufficiently large $n\in\N$. This was shown in \cite{Duan2005} when $x,y$ are probability vectors, but the proof extends directly to the more abstract setting considered here. The converse implication typically does not holds; Theorem 3 of \cite{Feng_et_al_2006} provides a recipe for deriving counter examples in a particular setting.

 \subsection{Jordan's Lemma}

In our earlier classical results in \cite{farooq2024,verhagen2025}, we were aided by the fact that tuples of classical probability vectors can be jointly decomposed into single-entry tuples. This enabled us to utilize Theorem \ref{thm:Fritz2022} in a simple manner directly in the classical case in finding the monotone homomorphisms and derivations since these maps can now be decomposed into simpler maps on smaller spaces. A similar decomposition is clearly impossible in the case of general pairs of quantum states. However, in some settings we may proceed in a way similar to the classical case. Namely, Jordan's Lemma \cite{Jordan1875} says that we may jointly decompose a pair of projections so that the components are of rank at most 1:
    
	\begin{lemma}[Jordan's Lemma]
		Let $\Hil$ be a finite-dimensional linear space and $A,B$ projections in $\Hil$ (i.e. two operators on $\Hil$ such that $A^2=A=A^*,B^2=B=B^*$). Then $A,B$ can be simultaneously block-diagonalized with blocks of dimension $1$ or $2$. In the case of dimension $2$, the corresponding blocks in $A,B$ have rank $1$. More precisely, there exists a decomposition $\Hil=\bigoplus_{i=1}^n\Hil_i$, for some $n\geq1$, where the $\Hil_i$ are at most $2$ dimensional linear spaces, such that 
		\begin{equation}\label{eq:Jordanform}
            A=\bigoplus_{i=1}^n A_i,\quad B=\bigoplus_{i=1}^n B_i, 
		\end{equation}
        where the $A_i$, $B_i$ are projections of rank at most $1$ in $\Hil_i$. 
	\end{lemma}

    See \cite{Bottcher2010} for a modern review of the theory of two projections. Jordan's Lemma lets us decompose a pair of flat states, i.e., two states $\rho=\tr[P]^{-1}P$ and $\sigma=\tr[Q]^{-1}Q$, where $P$ and $Q$ are projections, into at most rank-1 components. This motivates us to construct the smallest preordered semiring that contains all pairs of flat states, for which we can find all the monotone homomorphisms and derivations and hence derive our results on large-sample and catalytic majorization using Theorem \ref{thm:Fritz2022}. 

    \section{Analyzing the Semirings}\label{sec:technical}

    We will first define the preordered semirings applicable to our setting, and then derive the precise form that monotone homomorphisms and derivations have. These results are partly based on previous results on classical matrix majorization in \cite{farooq2024,verhagen2025}. 

    \subsection{Defining the Semirings}\label{sec:definingthesemirings}
    We consider the set $V$ consisting of all pairs of positive semi-definite matrices that are mutually block-diagonalizable into blocks of rank at most $1$. More precisely, $(A,B)\in V$ when $A,B$ are positive semi-definite operators on a finite-dimensional Hilbert space
    \begin{equation}
        \Hil=\bigoplus_{i=1}^{n}\Hil_i, 
    \end{equation}
    where $\Hil_1,\ldots,\Hil_n$, are $n\geq1$ Hilbert spaces (these spaces may all vary and also their number $n$ may vary), such that $A,B$ admit a decomposition  
    \begin{equation}\label{eq:Vpair}
        A=\bigoplus_{i=1}^{n}\ketbra{\alpha_i}{\alpha_i},\quad B=\bigoplus_{i=1}^{n}\ketbra{\beta_i}{\beta_i}, 
    \end{equation}
    where $\ket{\alpha_i},\ket{\beta_i}\in\Hil_i$ are vectors (possibly equal to the zero-vector). 
    
    We define the set of \emph{minimal restrictions} pairs $V_{\rm m.r.}\subsetneq V$ to be those $(A,B)\in V$ such that either $A$ and $B$ are both the $0$-operator in their associated Hilbert space, or $A$ and $B$ have some overlap in the sense that 
    \begin{equation}\label{eq:cscondition}
        \braket{\alpha_i}{\beta_i}\neq0\text{ for some }i\in[n]. 
    \end{equation}
    Also, we define the set of \emph{everywhere overlapping} pairs $V_{\rm e.o.}\subsetneq V_{\rm m.r.}$ to be those $(A,B)\in V_{\rm m.r.}$ such that, 
    for any $i\in[n]$, we have either 
    \begin{equation}
        \braket{\alpha_i}{\beta_i}\neq0 
    \end{equation}
    or 
    \begin{equation}
        \ket{\alpha_i}=0=\ket{\beta_i}. 
    \end{equation}
    Note that any pair of flat states (in particular, pair of pure states) that are at least partially overlapping is in $V$. 

    $V$ is closed under the $\boxplus$ and $\boxtimes$ operations, defined by 
    \begin{equation}
        (A,B)\boxplus(C,D):=(A\oplus C,B\oplus D),\quad(A,B)\boxtimes(C,D):=(A\otimes C,B\otimes D), 
    \end{equation}
    as one easily sees. We say that $(A,B)\in V$ with associated Hilbert space $\Hil$ majorizes $(C,D)\in V$ with associated Hilbert space $\Hil'$ and write $(A,B)\succeq(C,D)$ if and only if there exists a quantum channel $\Cha:\Lin(\Hil)\rightarrow\Lin(\Hil')$ such that $\Cha(A)=C$ and $\Cha(B)=D$. 

    We define the equivalence relation $\approx$ as follows: Let $(A,B),(A',B')\in V$ with associated Hilbert spaces $\Hil$, respectively $\Hil'$, and write $(A,B)\approx(A',B')$ if and only if they are equal up to a permutation of the component Hilbert spaces $\Hil_i$, adding/removing copies of $(0,0)$, or by transformations of the form 
    \begin{equation}
        V_i\ketbrar{\alpha_i}V_i^*=W_i\ketbrar{\alpha'_i}W_i^*,\quad V_i\ketbrar{\beta_i}V_i^*=W_i\ketbrar{\beta'_i}W_i^*, 
    \end{equation}
    where the $V_i,W_i$ are isometries from $\Hil_i,\Hil'_i$ into another Hilbert space. 
    For $(A,B)\in V$, we denote by $[(A,B)]$ the $\approx$-equivalence class of $(A,B)$. Define $S_{\rm m.r.}:=V_{\rm m.r.}/\!\approx$. The operations $\boxplus$ and $\boxtimes$ induce well-defined operations on $S_{\rm m.r.}$. These operations are now also commutative: for $\boxplus$, this follows from the freedom to permute the Hilbert spaces $\Hil_i$, and for $\boxtimes$, this can be seen by considering unitaries that swap two vectors in a tensor product. Also, the preorder $\succeq$ induces a preorder on $S_{\rm m.r.}$. In conclusion, $(S_{\rm m.r.},\boxplus,0:=[(0,0)],\boxtimes,1:=[(1,1)],\succeq)$ is a preordered semiring. We may also define the preordered semiring based on $S_{\rm e.o.}=V_{\rm e.o.}/\!\approx$ in the same way. In the sequel, we will often treat elements of $V_{\rm m.r.}$ or $V_{\rm e.o.}$ as elements of the semirings $S_{\rm m.r.}$ or $S_{\rm e.o.}$, i.e.,\ we usually ignore the $\approx$-equivalence classes, and simply write $(A,B)=(A',B')$ when $(A,B)\approx(A',B')$ for $(A,B),(A',B')\in V_{\rm m.r.}$ or $V_{\rm e.o.}$. We will also treat functions (such as the monotone homomorphisms) on the semirings as functions on $V_{\rm m.r.}$ or $V_{\rm e.o.}$; we simply assume that they are constant on each equivalence class. 

    In the rest of this paper, we will usually write a pair $(A,B)$ in $V_{\rm m.r.}$ or $V_{\rm e.o.}$ as follows: 
    \begin{equation}\label{eq:Vpair2}
        A=\bigoplus_{i=1}^{n}p_i\ketbra{\alpha_i}{\alpha_i},\quad B=\bigoplus_{i=1}^{n}q_i\ketbra{\beta_i}{\beta_i}, 
    \end{equation}
    where, for $i=1,\ldots,n$, the vector $\ket{\alpha_i}$ is normalized when $p_i>0$ and equal to the zero vector when $p_i=0$, and similarly for the $\ket{\beta_i}$, $q_i$. Note that when $(A,B)$ is a pair of quantum states, meaning that $\tr[A]=1=\tr[B]$, then $p_1+\ldots+p_n=1=q_1+\ldots+q_n$, i.e. the $p_i$ and the $q_i$ are both finite probability distributions. In this case, the pair $(A,B)$ is  essentially the same as the pair $(\rho,\sigma)$ in \eqref{eq:pairform}, but without the classical flags $\ketbrar{i}$ and using the direct product instead of the sum, since the component states $\ketbrar{\alpha_i}$ are now associated to different Hilbert spaces $\Hil_i$. One sees that such pairs are precisely those in $\mc F$, i.e. cq-states with pure components that have some overlap. 

    In \cite{farooq2024}, we discussed a semiring of $d$-tuples of vectors ($d\in\N$) with positive entries, where the preorder was defined by the existence of a stochastic map between the columns. In \cite{verhagen2025}, Section 3, we considered similar $d$-tuples, where entries are also allowed to equal zero. For the case $d=2$, our current semiring $S_{\rm e.o.}$ contains the former semiring, and the semiring $S_{\rm m.r.}$ contains the latter. More precisely, in both $S_{\rm e.o.}$ and $S_{\rm m.r.}$, the earlier studied semirings correspond to the subset of those $(A,B)$ in $S_{\rm e.o.}$ or $S_{\rm m.r.}$ where $A$ and $B$ commute. This means that we already know something about the monotone homomorphisms of $S_{\rm e.o.}$ and $S_{\rm m.r.}$ based on our earlier work, and we will use this knowledge in our further analysis. 

    We introduce some notation that will make the connection with the previously studied semirings in \cite{farooq2024,verhagen2025}, for $d=2$, more clear. Let $(A,B)$ be a pair as in \eqref{eq:Vpair2}. If it has a component $(p_i\ketbrar{\alpha_i},q_i\ketbrar{\beta_i})$ such that $|\braket{\alpha_i}{\beta_i}|=1$ for some $i$, then we may regard this as a ``classical'' component and we simply denote it as $(p_i,q_i)$, which is allowed by the equivalence relation defined on the semirings. Also, when $\braket{\alpha_i}{\beta_i}=0$, but $\ket{\alpha_i},\ket{\beta_i}\neq0$, we denote it as $(p_i,0)\boxplus(0,q_i)$. Similarly, when $\ket{\alpha_i}=0$, respectively $\ket{\beta_i}=0$, we denote the component as $(0,q_i)$, respectively $(p_i,0)$. As an example, see \eqref{eq:preparation} later on. 

       
   \subsection{Power Universal Elements}
    We will identify elements in the semirings $S_{\rm m.r.}$ and $S_{\rm e.o.}$ that are power universal. First, we present the following useful result due to Uhlmann \cite{Uhlmann1985}. 
    \begin{proposition}[Uhlmann]\label{propo:Uhlmann}
        Let $\ket{\alpha_1},\ket{\beta_1}\in\Hil_1$ and $\ket{\alpha_2},\ket{\beta_2}\in\Hil_2$ be normalized vectors in some Hilbert spaces $\Hil_1,\Hil_2$. Then 
        \begin{equation}
            (\,\ketbrar{\alpha_1},\ketbrar{\beta_1}\,)\succeq(\,\ketbrar{\alpha_2},\ketbrar{\beta_2}\,)\quad\Longleftrightarrow\quad|\braket{\alpha_1}{\beta_1}|\leq|\braket{\alpha_2}{\beta_2}|. 
        \end{equation}
    \end{proposition}

    Using this proposition, we first prove a characterization for certain elements in $S_{\rm e.o.}$ that are power universal. 

    \begin{proposition}\label{prop:puS}
    A pair $(A,B)$ of (normalized) quantum states of the form \eqref{eq:Vpair2} such that $0<|\braket{\alpha_i}{\beta_i}|<1$ for all $i=1,\ldots,n$ is a power universal of $S_{\rm e.o.}$.
    \end{proposition}

    Note that the above condition for the power universals of $S_{\rm e.o.}$ essentially means that all the components of the two cq-states are, in a sense, purely quantum.

    \begin{proof}
    Let us first consider a simpler case where we show that a pair $(\ketbrar{\alpha},\ketbrar{\beta})$ of pure states such that $0<|\braket{\alpha}{\beta}|<1$ is a power universal of $S_{\rm e.o.}$. To this end, let us fix $(A',B')\in S_{\rm e.o.}$ such that $\tr[A']=1=\tr[B']$, where
    \begin{equation}\label{eq:a'b'}
        A'=\bigoplus_{i=1}^{n'}p'_i\ketbrar{\alpha'_i},\quad B'=\bigoplus_{i=1}^{n'}q'_i\ketbrar{\beta'_i}.
    \end{equation}
    Since we are within $S_{\rm e.o.}$, we may assume that $p'_i,q'_i>0$ for all $i=1,\ldots,n'$. We next show that 
    \begin{equation}
        (\ketbrar{\alpha},\ketbrar{\beta})^{\boxtimes m}\succeq (A',B')
    \end{equation}    
    for $m\in\N$ sufficiently large. 

    Let us denote $F:=|\braket{\alpha}{\beta}|^2$. Let $m_0\in\N$ be such that $F^m p'_i<q'_i$ for all $m\geq m_0$ and $i=1,\ldots,n'$. Let $\ket{\alpha_m^\perp}$ be a unit vector in the 2-dimensional Hilbert space $\mc M_m$ spanned by $\ket{\alpha}^{\otimes m}$ and $\ket{\beta}^{\otimes m}$ which is perpendicular to $\ket{\alpha}^{\otimes m}$ for all $m\in\N$. It follows that we may define, for all $m\geq m_0$, the positive-operator-valued measure $E^m=(E^m_i)_{i=1}^{n'}$ on $\mc M_m$ through
    \begin{equation}
    E^m_i=p'_i\ketbrar{\alpha}^{\otimes m}+\frac{q'_i-F^mp'_i}{1-F^m}\ketbrar{\alpha_m^\perp}.
    \end{equation}
    In what follows, we write, e.g.,\ $|\alpha^{\otimes m}\rangle:=|\alpha\rangle^{\otimes m}$. Through a straightforward calculation, we find that
    \begin{equation}
    \left\langle\alpha^{\otimes m}\middle|E^m_i\middle|\alpha^{\otimes m}\right\rangle=p'_i,\qquad \left\langle\beta^{\otimes m}\middle|E^m_i\middle|\beta^{\otimes m}\right\rangle=q'_i.
    \end{equation}
    Let us define the unit vectors $\big|\tilde{\beta}_i^m\big\rangle=(q'_i)^{-1/2}(E^m_i)^{1/2}\ket{\beta}^{\otimes m}$. It is easily seen that
    \begin{equation}
    \big\langle\alpha^{\otimes m}\big|\tilde{\beta}_i^m\big\rangle
    =\sqrt{p'_i/q'_i}\braket{\alpha}{\beta}^m\to 0
    \end{equation}
    as $m\to\infty$. Since $|\braket{\alpha'_i}{\beta'_i}|>0$ for $i=1,\ldots,n'$, there is $m_1\in\N$, $m_1\geq m_0$, such that, for any $m\geq m_1$ and $i=1,\ldots,n'$, $|\big\langle\alpha^{\otimes m}\big|\tilde{\beta}_i^m\big\rangle|\leq|\braket{\alpha'_i}{\beta'_i}|$. Then, according to Proposition \ref{propo:Uhlmann}, for any $m\geq m_1$ and $i=1,\ldots,n'$, there is a channel $\mc C^m_i$ such that $\mc C^m_i(\ketbrar{\alpha}^{\otimes m})=\ketbrar{\alpha'_i}$ and $\mc C^m_i\big(\big|\tilde{\beta}_i^m\big\rangle\big\langle\tilde{\beta}_i^m\big|\big)=\ketbrar{\beta'_i}$. Thus, when $m\geq m_1$, we may define the channel $\mc C_m$ through
    \begin{equation}
    \mc C_m(\tau)=\bigoplus_{i=1}^{n'}\mc C_i^m\big((E^m_i)^{1/2}\tau(E_i^m)^{1/2}\big),
    \end{equation}
    and it is seen through a simple calculation that
    \begin{align}
    \mc C_m(\ketbrar{\alpha}^{\otimes m})&=\bigoplus_{i=1}^{n'}p'_i\ketbrar{\alpha'_i}=A',\\
    \mc C_m(\ketbrar{\beta}^{\otimes m})&=\bigoplus_{i=1}^{n'}q'_i\ketbrar{\beta'_i}=B'.
    \end{align}

    Let us next assume that the state pair $(A,B)$ is as in \eqref{eq:Vpair2} with $0<|\braket{\alpha_i}{\beta_i}|<1$ for $i=1,\ldots,n$. We are free to assume that $|\braket{\alpha_1}{\beta_1}|^2$ is the greatest of the fidelities $|\braket{\alpha_i}{\beta_i}|^2$. According to Proposition \ref{propo:Uhlmann}, we have channels $\mc D_i$ such that $\mc D_i(\ketbrar{\alpha_i})=\ketbrar{\alpha_1}$ and $\mc D_i(\ketbrar{\beta_i})=\ketbrar{\beta_1}$ for $i=1,\ldots,n$; for $i=1$ we may naturally choose the identity map $\mc D_1={\rm id}$. Let $P_i$ be the projection of the Hilbert space onto the $i$'th component in \eqref{eq:Vpair2} and define the channel $\mc D$ by $\mc D(\tau)=\sum_{i=1}^n\mc D_i(P_i\tau P_i)$. It follows that $\mc D(A)=\ketbrar{\alpha_1}$ and $\mc D(B)=\ketbrar{\beta_1}$, so that $(A,B)$ majorizes $(\ketbrar{\alpha_1},\ketbrar{\beta_1})$. Since, by assumption, $0<|\braket{\alpha_1}{\beta_1}|<1$, we know from what we have shown thus far that $(\ketbrar{\alpha_1},\ketbrar{\beta_1})^{\boxtimes m}\succeq (A',B')$ for any normalized element $(A',B')\in S_{\rm e.o.}$ for sufficiently large $m\in\N$. 
    Thus, also for any normalized $(A',B')\in S_{\rm e.o.}$,
    \begin{equation}\label{eq:apu2}
    (A,B)^{\boxtimes m}\succeq(\ketbrar{\alpha_1},\ketbrar{\beta_1})^{\boxtimes m}\succeq (A',B')
    \end{equation}
    for $m\in\N$ large enough. 

    We are finally able to show that the above $(A,B)$ is a power universal. For this, consider pairs $(A_1,B_1),(A_2,B_2)\in S_{\rm e.o.}$ such that $(A_1,B_1)\preceq(A_2,B_2)$. This clearly means especially that $\tr[A_1]=\tr[A_2]$ and $\tr[B_1]=\tr[B_2]$. Define $A':=\tr[A_2]^{-1}A_2$ and $B':=\tr[B_2]^{-1}B_2$. Since we have already proven that, for $m\in\N$ large enough, we have \eqref{eq:apu2}, we now have
    \begin{align}
    (A_1,B_1)\boxtimes(A,B)^{\boxtimes m}&=(A_1\otimes A^{\otimes m},B_1\otimes B^{\otimes m})\\
    &\succeq\big(\tr[A_1]A^{\otimes m},\tr[B_1]B^{\otimes m}\big)\\
    &=\big(\tr[A_2]A^{\otimes m},\tr[B_2]B^{\otimes m}\big)\\
    &\succeq\big(\tr[A_2]A',\tr[B_2]B'\big)=(A_2,B_2)
    \end{align}
    for $m\in\N$ large enough. This shows that $(A,B)$ is a power universal.
    \end{proof}

    Using the above result, we may easily characterize certain power universals of $S_{\rm m.r.}$. These are essentially of the same type as the power universals in $S_{\rm e.o.}$, but with an additional orthogonal component. 

    \begin{proposition}\label{prop:pumr}
    A pair $(A,B)$ of (normalized) quantum states of the form \eqref{eq:Vpair2} such that $|\braket{\alpha_i}{\beta_i}|<1$ for all $i=1,\ldots,n$, and $\braket{\alpha_i}{\beta_i}=0$ for some $i$ with $\ket{\alpha_i},\ket{\beta_i}\neq0$, is a power universal of $S_{\rm m.r.}$.
    \end{proposition}

    \begin{proof}
    Let $(A,B)$ be a pair as stated in the claim. Possibly by reordering the blocks (i.e. component spaces $\Hil_i$), we may assume that there are $\ell$ with $1\leq\ell<n$, unit vectors $\ket{\alpha_i}$ and $\ket{\beta_i}$ such that $0<|\braket{\alpha_i}{\beta_i}|<1$, $i=1,\ldots,\ell$, positive probabilities $p_1,\ldots,p_\ell$ and $q_1,\ldots,q_\ell$ with $P:=p_1+\cdots+p_\ell<1$ and $Q:=q_1+\cdots+q_\ell<1$, and non-zero vectors $p,q\in\R_+^{n-\ell}$, $p=(p_{\ell+1},\ldots,p_{n})$, $q=(q_{\ell+1},\ldots,q_n)$, with $\|p\|_1+P=1=\|q\|_1+Q$ such that
    \begin{align}
    A&=\left(\bigoplus_{i=1}^\ell p_i\ketbrar{\alpha_i}\right)\oplus p,\\
    B&=\left(\bigoplus_{i=1}^\ell q_i\ketbrar{\beta_i}\right)\oplus q.
    \end{align}
    Since for all $i>\ell$ we have $\braket{\alpha_i}{\beta_i}=0$ and $\ket{\alpha_i},\ket{\beta_i}\neq0$, we deduce that 
    \begin{equation}
    {\rm supp}\,p\cap{\rm supp}\,q=\emptyset,
    \end{equation}
    i.e.,\ $p$ and $q$ are disjoint.

    It follows that, for each $m\in\N$, 
    \begin{align}
    A^{\otimes m}&=\left(\bigoplus_{i_1,\ldots,i_m=1}^\ell p_{i_1}\cdots p_{i_m}\ketbrar{\alpha_{i_1}}\otimes\cdots\otimes\ketbrar{\alpha_{i_m}}\right)\oplus A^{(m)},\\
    B^{\otimes m}&=\left(\bigoplus_{i_1,\ldots,i_m=1}^\ell q_{i_1}\cdots q_{i_m}\ketbrar{\beta_{i_1}}\otimes\cdots\otimes\ketbrar{\beta_{i_m}}\right)\oplus B^{(m)}
    \end{align}
    for some operators $A^{(m)}$ and $B^{(m)}$ where $A^{(m)}$ is a tensor product of $p$ (seen as a diagonal operator) with some other operator and $B^{(m)}$ is a tensor product of $q$ with some other operator. Since $p$ and $q$ have disjoint supports, $A^{(m)}$ and $B^{(m)}$ commute, so we may identify them with some sequences $p^{(m)}$ and $q^{(m)}$ (classical blocks) of non-negative numbers with disjoint supports. We may assume that the block $i=1$ has the highest fidelity among $|\braket{\alpha_i}{\beta_i}|^2$ for $i=1,\ldots,\ell$.
    Then, by Proposition \ref{propo:Uhlmann}, there are channels $\mc C_{i_1,\ldots,i_m}$ such that
    \begin{align}
    \mc C_{i_1,\ldots,i_m}\big(\ketbrar{\alpha_{i_1}}\otimes\cdots\otimes\ketbrar{\alpha_{i_m}}\big)&=\ketbrar{\alpha_1}^{\otimes m},\\
    \mc C_{i_1,\ldots,i_m}\big(\ketbrar{\beta_{i_1}}\otimes\cdots\otimes\ketbrar{\beta_{i_m}}\big)&=\ketbrar{\beta_1}^{\otimes m}.
    \end{align}
    Let $P_{i_1,\ldots,i_m}$ be the projections onto the $\ell^m$ first blocks and $Q_j$, $j=\ell^m+1,\ldots,n^m$ be the projections onto the last (classical) blocks. Let us also denote by $\big({\rm supp}\,p^{(m)}\big)^c$ the complement of the support of $p^{(m)}$ among the $n^m-\ell^m$ final classical blocks. Naturally, this complement contains the support of $q^{(m)}$. Let us set up the quantum channel $\mc C_m$
    \begin{align}
    \mc C_m(R)=&\sum_{i_1,\ldots,i_m=1}^\ell \mc C_{i_1,\ldots,i_m}(P_{i_1,\ldots,i_m}RP_{i_1,\ldots,i_m})\\
    &\oplus\,\left(\sum_{j\in{\rm supp}\,p^{(m)}}\tr[Q_jRQ_j]\right)\,\oplus\,\left(\sum_{k\in\big({\rm supp}\,p^{(m)}\big)^c}\tr[Q_k RQ_k]\right).
    \end{align}
    We now have 
    \begin{align}
    \mc C_m(A^{\otimes m})&=P^m\ketbrar{\alpha_1}^{\otimes m}\oplus(1-P^m)\oplus 0=:A_m,\\
    \mc C_m(B^{\otimes m})&=Q^m\ketbrar{\beta_1}^{\otimes m}\oplus 0\oplus(1-Q^m)=:B_m,
    \end{align}
    so that $(A,B)^{\boxtimes m}\succeq (A_m,B_m)$.

    Let $(A',B')$ be a pair of operators within $S_{\rm m.r.}$ which we intend to majorize with a high enough power of the input $(A,B)$. Let us assume that this pair has the form of \eqref{eq:a'b'}, except that now not all $p'_i$ and $q'_i$ need to be non-zero. However, there needs to be a block $i_0$ such that $p'_{i_0},q'_{i_0}>0$ and $\braket{\alpha'_{i_0}}{\beta'_{i_0}}\neq0$. By reordering, we may assume that $i_0=1$. Let us now upper bound this target pair $(A',B')$. We may write $A'=A_1\oplus A_2$ and $B'=B_1\oplus B_2$ where $(A_1,B_1)\in S_{\rm e.o.}$ and $A_2$ and $B_2$ have blocks that are completely disjoint. Let $s:=\tr[A_1]$ and $t:=\tr[B_1]$. We now claim that $(\tilde{A},\tilde{B})\succeq (A',B')$ where
    \begin{align}
    \tilde{A}&=s\ketbrar{\alpha}\oplus(1-s)\oplus 0,\\
    \tilde{B}&=t\ketbrar{\beta}\oplus 0\oplus(1-t)
    \end{align}
    where $\ket{\alpha}$ and $\ket{\beta}$ are some non-orthogonal unit vectors. Indeed, from the proof Proposition \ref{prop:puS}, we know that (when $|\braket{\alpha}{\beta}|$ is sufficiently small) \sloppy{${(\ketbrar{\alpha},\ketbrar{\beta})\succeq (\tr[A_1]^{-1}A_1,\tr[B_1]^{-1}B_1)}$}, and from the remaining orthogonal part \sloppy{${((1-s)\oplus0,0\oplus(1-t))}$} of $(\tilde{A},\tilde{B})$ we can easily find a channel that creates $(A_2,B_2)$. We next show that there is $m\in\N$ such that $(A,B)^{\boxtimes m}\succeq (A',B')$ by showing the intermediary $(A_m,B_m)\succeq(\tilde{A},\tilde{B})$ for sufficiently large $m\in\N$.

    Let us denote the 2-dimensional space spanned by $\ket{\alpha_1}^{\otimes m}$ and $\ket{\beta_1}^{\otimes m}$ by $\mc M_m$ and the projection of $\mc M_m\oplus\C\oplus\C$ onto $\mc M_m$ by $R_m$. The projection orthogonal to this we denote by $R_m^\perp$. Using again Proposition \ref{propo:Uhlmann}, we find that, for $m\in\N$ large enough, there is a channel $\mc D_m$ such that $\mc D_m(\ketbrar{\alpha_1}^{\otimes m})=\ketbrar{\alpha}$ and $\mc D_m(\ketbrar{\beta_1}^{\otimes m})=\ketbrar{\beta}$. Let us assume that $m\in\N$ is sufficiently large so that $P^m<s$ and $Q^m<t$. 
    We may now define a channel $\mc E_m$ on the last two blocks of $(A_m,B_m)$ (which together can be regarded as a pair of qubits) such that
    \begin{align}
    \mc E_m(1\oplus 0)&=\frac{s-P^m}{1-P^m}\ketbrar{\alpha}\oplus\frac{1-s}{1-P^m}\oplus0,\\
    \mc E_m(0\oplus 1)&=\frac{t-Q^m}{1-Q^m}\ketbrar{\beta}\oplus0\oplus\frac{1-t}{1-Q^m}.
    \end{align}
    Define the channel $\mc C_m$ through
    \begin{equation}
    \mc C_m(\tau)=\mc D_m(R_m\tau R_m)+\mc E_m(R_m^\perp\tau R_M^\perp).
    \end{equation}
    It now follows from a straightforward calculation that $\mc C_m(A_m)=\tilde{A}$ and $\mc C_m(B_m)=\tilde{B}$. Thus, we may deduce that $(A,B)^{\boxtimes m}\succeq(A',B')$ for $m\in\N$ large enough. The fact that $(A,B)$ is a power universal now follows through the same simple logic as in the end of the proof of Proposition \ref{prop:puS}.
    \end{proof}

    \subsection{Finding the Monotone Homomorphisms and Derivations}\label{sec:spectrum}
    We will identify all non-degenerate monotone homomorphisms with trivial kernel and non-zero monotone derivations associated with the semiring $S_{\rm m.r.}$, see Proposition \ref{prop:monhom}. We will also derive partial results for $S_{\rm e.o.}$, summarized in Remark \ref{rmk:Seo}. According to Theorem \ref{thm:Fritz2022}, sufficient conditions for large-sample and catalytic majorization can be stated in terms of these homomorphisms and derivations. In the following calculations and proofs we will repeatedly use vectors
    \begin{equation}\label{eq:anglevectors}
    |\psi(\theta)\rangle:=\cos{\theta}\,|0\rangle+\sin{\theta}\,|1\rangle\in\C^2,\qquad\theta\in[0,2\pi),
    \end{equation}
    where $\{|0\rangle,|1\rangle\}$ is some fixed orthonormal basis of $\C^2$. The next Proposition shows that, compared with our earlier classical results in \cite{farooq2024,verhagen2025}, a new parameter $z$ appears due to the quantum nature of the pairs of states. 

    
    \begin{proposition}\label{prop:z}
        Let $\Phi:S_{*}\rightarrow\K$ be a monotone homomorphism, where $S_{*}\in\{S_{\rm m.r.},S_{\rm e.o.}\}$ and $\K\in\{\R_+,\R_+^{\rm op},\T\R_+,\T\R_+^{\rm op}\}$. There exists $z\in\R$ such that 
        \begin{equation}
            \Phi(\ketbrar{\alpha},\ketbrar{\beta})=|\braket{\alpha}{\beta}|^{2z}
        \end{equation}
        for all normalized $\ket{\alpha},\ket{\beta}\in\Hil$ such that $\braket{\alpha}{\beta}\neq0$, and $\Hil$ is a Hilbert space. 
    \end{proposition}

    \begin{proof}
        Fix orthonormal vectors $\ket{0}$ and $\ket{1}$ that span the subspace $\mc M$ where $\ket{\alpha}$ and $\ket{\beta}$ reside. For any normalized $|\alpha^\perp\rangle\in\mc M$ perpendicular to $\ket{\alpha}$ we have 
        \begin{equation}
            |\braket{\alpha}{\beta}|^2+|\langle\alpha^\perp|\beta\rangle|^2=1. 
        \end{equation}
        Hence 
        \begin{equation}
            |\braket{\alpha}{\beta}|=\cos\theta,\quad|\langle\alpha^\perp|\beta\rangle|=\sin\theta,
        \end{equation}
        for some $\theta\in[0,\pi/2]$. By multiplying with an appropriate phase, we may choose $|\alpha^\perp\rangle$ such that 
        \begin{equation}
            \braket{\alpha}{\beta}\langle\beta|\alpha^\perp\rangle=\cos\theta\sin\theta. 
        \end{equation}
        Via the unitary operation defined by $\ket{\alpha}\mapsto\ket{0},|\alpha^\perp\rangle\mapsto\ket{1}$, the pair $(\ketbra{\alpha}{\alpha},\ketbra{\beta}{\beta})$ is equivalent (according to the equivalence relation $\approx$ defined in Section \ref{sec:definingthesemirings}) to the pair 
        \begin{equation}\label{eq:pairstandardform}
            (\ketbra{0}{0},\cos^2\theta\ketbra{0}{0}+\cos\theta\sin\theta\ketbra{0}{1}+\cos\theta\sin\theta\ketbra{1}{0}+\sin^2\theta\ketbra{1}{1}). 
        \end{equation}
        From this, we see that $\Phi(\ketbrar{\alpha},\ketbrar{\beta})$ only depends on the fidelity $\cos^2\theta=|\braket{\alpha}{\beta}|^2$. Therefore, we can define $\varphi:(0,1]\rightarrow\R$, 
        \begin{equation}
            \varphi(x):=\Phi(\ketbrar{\alpha},\ketbrar{\beta}), 
        \end{equation}
        where $\ket{\alpha},\ket{\beta}$ are any normalized vectors such that $|\braket{\alpha}{\beta}|^2=x$. Note that 
        \begin{equation}
            (\ketbrar{\alpha_1},\ketbrar{\beta_1})\boxtimes(\ketbrar{\alpha_2},\ketbrar{\beta_2})=(\ket{\alpha_1}\otimes\ket{\alpha_2}\bra{\alpha_1}\otimes\bra{\alpha_2},\ket{\beta_1}\otimes\ket{\beta_2}\bra{\beta_1}\otimes\bra{\beta_2}) 
        \end{equation}
        and 
        \begin{equation}
            |(\bra{\alpha_1}\otimes\bra{\alpha_2})(\ket{\beta_1}\otimes\ket{\beta_2})|^2=|\braket{\alpha_1}{\beta_1}|^2|\braket{\alpha_2}{\beta_2}|^2. 
        \end{equation}
        Thus, by multiplicativity of $\Phi$, the function $\varphi$ is multiplicative as well, i.e. 
        \begin{equation}
            \varphi(xy)=\varphi(x)\varphi(y)\text{ for all }x,y\in(0,1]. 
        \end{equation}    
        Using Proposition \ref{propo:Uhlmann}, for any $\ket{\alpha},\ket{\beta}$ such that 
        \begin{equation}
            |\braket{\alpha}{\beta}|^2\geq\frac{1}{2}=|\braket{0}{\psi(\pi/4)}|^2, 
        \end{equation}
        where $|\psi(\pi/4)\rangle$ is defined as in \eqref{eq:anglevectors}, we have 
        \begin{equation}
            (1,1)\preceq(\ketbrar{\alpha},\ketbrar{\beta})\preceq(\ketbrar{0},\ketbrar{\psi(\pi/4)}). 
        \end{equation}
        Hence for all $x\in[1/2,1]$, we have $\varphi(x)\leq1$ if $\K=\R_+^{\rm op},\T\R_+^{\rm op}$, and $\varphi(x)\leq\varphi(1/2)$ if $\K=\R_+,\T\R_+$. Therefore, $\varphi$ is bounded from above on a set of positive measure, which establishes that $\varphi$ is a regular solution of the multiplicative version of Cauchy's Functional Equation \cite[Chapter 3, Proposition 6]{Aczel1989}: there exists $z\in\R$ such that $\varphi(x)=x^z$ or $\varphi$ vanishes everywhere. The latter case is impossible since $\varphi(1)=\Phi(1,1)=1$ by the properties of a homomorphism. 
    \end{proof}

    \begin{remark}\label{rmk:classical}
    Let us summarize some results that were established in \cite{verhagen2025} on the possible monotone homomorphisms and derivations on commuting pairs $(A,B)$ in $S_{\rm m.r.}$, i.e.,\ a pair $(A,B)$ such as in \eqref{eq:Vpair2} where the component spaces $\Hil_i\simeq\C$. We will use these results extensively in the analysis of our current semiring $S_{\rm m.r.}$. 
    
    The commuting pairs in $S_{\rm m.r.}$ correspond precisely to the elements in the semiring denoted in \cite{verhagen2025} as $S^2$, the \emph{minimal restrictions matrix majorization semiring of length $2$}. Hence any non-degenerate monotone homomorphism or monotone derivation on $S_{\rm m.r.}$, when restricted to these commuting pairs, must correspond to a (possibly degenerate) monotone homomorphism or monotone derivation on $S^2$. 
    
    We now list these homorphisms and derivations. Let $\Phi:S_{\rm m.r.}\rightarrow\K$, $\K\in\{\R_+,\R_+^{\rm op},\T\R_+,\T\R_+^{\rm op}\}$, be a monotone homomorphism. When $\K=\R_+^{\rm op}$, $\Phi$ restricted to a commuting pair $(A,B)$ takes the form 
    \begin{equation}\label{eq:Phitemperate}
    \Phi(A,B)=\sum_{i=1}^n p_i^\alpha q_i^{1-\alpha}
    \end{equation}
    for some $\alpha\in(0,1)$, or 
    \begin{equation}\label{eq:Phicharacter}
    \Phi(A,B)=\sum_{i:p_i>0} q_i,\quad\text{or}\quad\Phi(A,B)=\sum_{i:q_i>0} p_i, 
    \end{equation}
    or $\Phi$ restricts to one of the two degenerate homomorphisms $\Phi(A,B)=\tr[A]$ or $\tr[B]$.\footnote{Note that, using terminology introduced in \cite{verhagen2025}, the two cases in \eqref{eq:Phicharacter} correspond to the \emph{character} $C=\{1,2\}$, i.e. only those $i$ for which $p_i>0$ and $q_i>0$ are summed over.} When $\K=\R_+$, $\Phi$ always restricts to one of these latter two degenerate cases. When $\mb K=\T\R_+$ or $\T\R_+^{\rm op}$, the only form $\Phi$ can take on a commuting pair is the highly degenerate $\Phi(A,B)=1$ when $(A,B)\neq 0$ and $\Phi(0,0)=0$. (However, it will turn out that the new parameter $z$ in Proposition \ref{prop:z} makes the extension of this homomorphism to $S_{\rm m.r.}$ \emph{non}-degenerate in the case $\K=\T\R_+^{\rm op}$.) Finally, the derivations corresponding to the two degenerate cases $\Phi(A,B)=\tr[A]$ or $\tr[B]$, up to interchangeability (see Footnote \ref{ftn:interchangeability} on page \pageref{ftn:interchangeability}), must vanish on commuting pairs.
    \end{remark}

    Using the additivity and multiplicativity of any homomorphism $\Phi:S_{\rm m.r.}\rightarrow\K$ in the cases $\K=\R_+,\R_+^{\rm op}$, we can decompose $\Phi(A,B)$, where $(A,B)\in S_{\rm m.r.}$ as in \eqref{eq:Vpair2}, as follows: 
    \begin{equation}\label{eq:decompositiontemperate}
        \Phi(A,B)=\sum_{i=1}^n\Phi(p_i,q_i)\Phi(\ketbrar{\alpha_i},\ketbrar{\beta_i}). 
    \end{equation}
    And similarly for the cases $\K=\T\R_+,\T\R_+^{\rm op}$: 
    \begin{equation}\label{eq:decompositiontropical}
        \Phi(A,B)=\max_{i=1,\ldots,n}\Phi(p_i,q_i)\Phi(\ketbrar{\alpha_i},\ketbrar{\beta_i}). 
    \end{equation}

    Note that if $\ket{\alpha_i},\ket{\beta_i}$ are orthogonal or exactly one of them is the zero-vector, then $(\ketbrar{\alpha_i},\ketbrar{\beta_i})$ is not part of $S_{\rm m.r.}$. Similarly for $(p_i,q_i)$ if either $p_i$ or $q_i$ is zero. Hence, some terms in the decompositions \eqref{eq:decompositiontemperate} and \eqref{eq:decompositiontropical} are not well-defined. We will address this issue in the first part of the proof of Proposition \ref{prop:monhom}. 

    By Proposition \ref{prop:z}, we know the value of all the ``quantum'' terms $\Phi(\ketbrar{\alpha_i},\ketbrar{\beta_i})$ in \eqref{eq:decompositiontemperate} and \eqref{eq:decompositiontropical}. Also, from the results on commuting pairs as summarized in Remark \ref{rmk:classical}, we know the value of all the ``classical'' terms $\Phi(p_i,q_i)$, as well. Therefore, the value of $\Phi(A,B)$ for any $(A,B)\in S_{\rm m.r.}$ follows. What remains is to determine exactly which of these homomorphisms are non-degenerate and monotone, which will yield the allowed range of the parameters $\alpha$ and $z$, and investigating the existence of non-zero monotone derivations. Our results are contained in the following Proposition. 
    


    \begin{proposition}\label{prop:monhom}
    The non-degenerate monotone homomorphisms $\Phi:S_{\rm m.r.}\to\mb K$, where $\K\in\{\R_+,\R_+^{\rm op},\T\R_+,\T\R_+^{\rm op}\}$, with trivial kernel are exactly the following: 
    \begin{itemize}
    \item $\mb K=\R_+^{\rm op}$: $\Phi$ is of the form 
    \begin{equation}\label{eq:temperate}
        \Phi_{\alpha,z}(A,B):=\sum_{i=1}^n(p_i)^\alpha(q_i)^{1-\alpha}|\braket{\alpha_i}{\beta_i}|^{2z}, 
    \end{equation}    
    for any $\alpha\in[0,1]$ and $z\geq\max\{\alpha,1-\alpha\}$. 
    \item $\mb K=\T\R_+^{\rm op}$: $\Phi$ is of the form 
    \begin{equation}\label{eq:tropicalz}
        \Phi_z^\T(A,B):=\max_{i=1,\ldots,n}|\braket{\alpha_i}{\beta_i}|^{2z}, 
    \end{equation}    
    for any $z>0$. 
    \item $\mb K\in\{\R_+,\T\R_+\}$: There are no non-degenerate monotone homomorphisms with trivial kernel. 
    \end{itemize}
    Moreover, there are no non-zero monotone derivations.
    \end{proposition}

    \begin{remark}    
        Note that when $\alpha=0$ or $1$, it may occur that the expression \eqref{eq:temperate} contains a factor $0^0$ when $p_i=0$ or $q_i=0$ for some $i$. However, since in this case also the corresponding $\ket{\alpha_i}$ or $\ket{\beta_i}$ equals the zero vector, there is no ambiguity: the entire expression $(p_i)^\alpha(q_i)^{1-\alpha}|\braket{\alpha_i}{\beta_i}|^{2z}$ should be considered to be equal to $0$. 

        One checks that for pairs $(A,B)\in S_{\rm m.r.}$, we have 
        \begin{equation}
            \Phi_{\alpha,z}(A,B)=\tr\left[\left(B^{\frac{1-\alpha}{2z}}A^{\frac{\alpha}{z}}B^{\frac{1-\alpha}{2z}}\right)^z\right], 
        \end{equation}
        where the right-hand side is the expression appearing inside the logarithm in the definition of the $\alpha$-$z$ relative entropy in \eqref{eq:alphazdivergence}. 
    
        For two pairs $(A,B),(A',B')\in S_{\rm m.r.}$, note that $\Phi_z^\T(A,B)>\Phi_z^\T(A',B')$ for all $z>0$ if and only if this inequality holds for one such $z$. Since the conditions for majorization given by Theorem \ref{thm:Fritz2022} are stated in terms of such inequalities, it is sufficient for our purposes to consider only $\Phi_z^\T$ for the choice $z=1$. Hence, we define 
        \begin{equation}\label{eq:tropical}
            \Phi^\T(A,B):=\Phi_1^\T(A,B)=\max_{i=1,\ldots,n}|\braket{\alpha_i}{\beta_i}|^{2}=\|({\rm supp}\,A)({\rm supp}\,B)\|_\infty^2. 
        \end{equation}
        Here, the right-hand side gives an alternative expression for the homomorphism, where $\|\cdot\|_\infty$ is the operator norm and, e.g.,\ ${\rm supp}\,A$ is the support projection of the positive semidefinite $A$. 
    \end{remark}

    \begin{proof}[Proof of Proposition \ref{prop:monhom}]
    Before proceeding to the characterization of all relevant monotone homomorphisms and derivations, we will first show how to interpret the decompositions \eqref{eq:decompositiontemperate} and \eqref{eq:decompositiontropical} when there are $i$ such that $\ket{\alpha_i},\ket{\beta_i}$ are orthogonal or exactly one of them is the zero-vector. 
    
    Fix a monotone homomorphism $\Phi:S_{\rm m.r.}\rightarrow\K$, where $\K\in\{\R_+,\R_+^{\rm op},\allowbreak\T\R_+,\T\R_+^{\rm op}\}$, and let $(A,B)\in S_{\rm m.r.}$. Note that we may assume that there exists a representation as in \eqref{eq:Vpair2} such that there is no $j$ for which $\braket{\alpha_j}{\beta_j}=0$ and $\ket{\alpha_j},\ket{\beta_j}\neq0$. Indeed, if there is such $j$, then the pair $(p_j\ketbrar{\alpha_j},q_j\ketbrar{\beta_j})$ can be replaced by $(p_j\oplus0,0\oplus q_j)$ under the equivalence relation $\approx$. 
    
    Due to \eqref{eq:cscondition}, there exists $k\in[n]$ such that $|\braket{\alpha_k}{\beta_k}|\neq0$. Assume there is a $j_1$ such that $\ket{\alpha_{j_1}}\neq0$ and $\ket{\beta_{j_1}}=0$, or a $j_2$ such that $\ket{\alpha_{j_2}}=0$ and $\ket{\beta_{j_2}}\neq0$. In the following, we will show that $\Phi(A,B)$ may be expressed as \eqref{eq:decompositiontemperate} or \eqref{eq:decompositiontropical}, where we set 
    \begin{equation}\label{eq:define0trA}
        \Phi(p_{j_1}\ketbrar{\alpha_{j_1}},0)=p_{j_1}\Phi(\ketbrar{\alpha_k},\ketbrar{\beta_k})\text{ and }\Phi(0,q_{j_2}\ketbrar{\beta_{j_2}})=0
    \end{equation}
    if $\K=\R_+$ or $\R_+^{\rm op}$ and $\Phi$ restricts to $\tr[A']$ on commuting pairs $(A',B')$, 
    \begin{equation}\label{eq:define0trB}
        \Phi(p_{j_1}\ketbrar{\alpha_{j_1}},0)=0\text{ and }\Phi(0,q_{j_2}\ketbrar{\beta_{j_2}})=q_{j_2}\Phi(\ketbrar{\alpha_k},\ketbrar{\beta_k})
    \end{equation}
    if $\K=\R_+$ or $\R_+^{\rm op}$ and $\Phi$ restricts to $\tr[B']$ on commuting pairs $(A',B')$, and 
    \begin{equation}\label{eq:define0otherwise}
        \Phi(p_{j_1}\ketbrar{\alpha_{j_1}},0)=0\text{ and }\Phi(0,q_{j_2}\ketbrar{\beta_{j_2}})=0
    \end{equation}
    otherwise. Note that, although \eqref{eq:define0trA}, \eqref{eq:define0trB} seem to depend on the choice of $k$ due to the appearance of the term $\Phi(\ketbrar{\alpha_k},\ketbrar{\beta_k})$, we will show later that $z=0$ when $\Phi$ restricts to $\tr[A']$ or $\tr[B']$ on commuting pairs $(A',B')$, hence $\Phi(\ketbrar{\alpha_k},\ketbrar{\beta_k})=1$, independent of $k$. 
    
    Assume there is a $j_1$ such that $\ket{\alpha_{j_1}}\neq0$ and $\ket{\beta_{j_1}}=0$; the analysis of the case of $j_2$ such that $\ket{\alpha_{j_2}}=0$ and $\ket{\beta_{j_2}}\neq0$ is analogous. Denote by $\substack{\succeq\\\preceq}$ majorization in both directions. There exist a quantum channel that splits block $k$, and another channel that merges two blocks, which implies 
    \begin{align}
        \big(p_k\ketbrar{\alpha_k},q_k\ketbrar{\beta_k}\big)\substack{\succeq\\\preceq}&\big(p_k/2\ketbrar{\alpha_k},q_k/2\ketbrar{\beta_k}\big)\\
        &\boxplus\big(p_k/2\ketbrar{\alpha_k},q_k/2\ketbrar{\beta_k}\big). 
    \end{align}
    Denote by $(\tilde A,\tilde B)\in S_{\rm m.r.}$ the pair $(A,B)$ with blocks $j_1$ and $k$ removed. Using the majorization relation above, we have 
    \begin{align}
        \big(A,B\big)\substack{\succeq\\\preceq}&\big(p_k/2\ketbrar{\alpha_k}\oplus \tilde A,q_k/2\ketbrar{\beta_k}\oplus\tilde B\big)\label{eq:split1}\\
        &\boxplus\big(p_k/2\ketbrar{\alpha_k}\oplus p_{j_1}\ketbrar{\alpha_{j_1}},q_k/2\ketbrar{\beta_k}\oplus0\big).\label{eq:split2}
    \end{align}
    Since $\ket{\beta_{j_1}}=0$, we may assume that $\ket{\alpha_{j_1}}=\ket{\alpha_k}$ under the equivalence relation $\approx$. Then, notice that 
    \begin{align}
        &\big(p_k/2\ketbrar{\alpha_k}\oplus p_{j_1}\ketbrar{\alpha_{j_1}},q_k/2\ketbrar{\beta_k}\oplus0\big)\label{eq:specialblock1}\\
        &\quad\quad=\big(\ketbrar{\alpha_k},\ketbrar{\beta_k}\big)\boxtimes\big(p_k/2\oplus p_{j_1},q_k/2\oplus0\big).\label{eq:specialblock2}
    \end{align}
    
    Now assume $\K=\R_+$ or $\R_+^{\rm op}$. Assume further that the restriction of $\Phi$ on commuting pairs $(A',B')$ does not equal $\tr[A']$. Considering the remaining possible restrictions (see Remark \ref{rmk:classical}), it follows from \eqref{eq:specialblock1}, \eqref{eq:specialblock2} that 
    \begin{align}
        &\Phi\big(p_k/2\ketbrar{\alpha_k}\oplus p_{j_1}\ketbrar{\alpha_{j_1}},q_k/2\ketbrar{\beta_k}\oplus0\big)\\
        &\quad\quad=\Phi\big(\ketbrar{\alpha_k},\ketbrar{\beta_k}\big)\Phi\big(p_k/2\oplus p_{j_1},q_k/2\oplus0\big)\\
        &\quad\quad=\Phi\big(\ketbrar{\alpha_k},\ketbrar{\beta_k}\big)\Phi\big(p_k/2,q_k/2\big)\\
        &\quad\quad=\Phi\big(p_k/2\ketbrar{\alpha_k},q_k/2\ketbrar{\beta_k}\big). 
    \end{align}
    Using this equality and \eqref{eq:split1}, \eqref{eq:split2}, via splitting and merging of blocks, we find 
    \begin{align}
        \Phi(A,B)&=\Phi\big(p_k/2\ketbrar{\alpha_k}\oplus\tilde A,q_k/2\ketbrar{\beta_k}\oplus\tilde B\big)\\
        &\quad+\Phi\big(p_k/2\ketbrar{\alpha_k},q_k/2\ketbrar{\beta_k}\big)\\        
        &=\Phi\big(p_k\ketbrar{\alpha_k}\oplus\tilde A,q_k\ketbrar{\beta_k}\oplus\tilde B\big).
    \end{align}
    Note that $(p_k\ketbrar{\alpha_k}\oplus\tilde A,q_k\ketbrar{\beta_k}\oplus\tilde B)$ equals $(A,B)$ with the pair $(p_{j_1}\ketbrar{\alpha_{j_1}},0)$ removed. Hence, we may set $\Phi(p_{j_1}\ketbrar{\alpha_{j_1}},0)=0$. Assume now that $\Phi$ restricts to $\tr[A']$ on commuting pairs $(A',B')$. Then it follows from \eqref{eq:specialblock1}, \eqref{eq:specialblock2} that 
    \begin{align}
        &\Phi\big(p_k/2\ketbrar{\alpha_k}\oplus p_{j_1}\ketbrar{\alpha_{j_1}},q_k/2\ketbrar{\beta_k}\oplus0\big)\\
        &\quad\quad=\Phi\big(\ketbrar{\alpha_k},\ketbrar{\beta_k}\big)(p_k/2+p_{j_1})\\
        &\quad\quad=\Phi\big(\ketbrar{\alpha_k},\ketbrar{\beta_k}\big)(\Phi\big(p_k/2,q_k/2\big)+p_{j_1})\\        &\quad\quad=\Phi\big(p_k/2\ketbrar{\alpha_k},q_k/2\ketbrar{\beta_k}\big)+p_{j_1}\Phi\big(\ketbrar{\alpha_k},\ketbrar{\beta_k}\big). 
    \end{align}
    Now it follows from \eqref{eq:split1}, \eqref{eq:split2} that 
    \begin{equation}
        \Phi(A,B)=\Phi\big(p_k\ketbrar{\alpha_k}\oplus\tilde A,q_k\ketbrar{\beta_k}\oplus\tilde B\big)+p_{j_1}\Phi(\ketbrar{\alpha_k},\ketbrar{\beta_k}), 
    \end{equation}
    allowing us to set $\Phi(p_{j_1}\ketbrar{\alpha_{j_1}},0)=p_{j_1}\Phi(\ketbrar{\alpha_k},\ketbrar{\beta_k})$. Analogously one can show that, when there is $j_2$ such that $\ket{\alpha_{j_2}}=0$ and $\ket{\beta_{j_2}}\neq0$, and $\Phi$ does not restrict to $\tr[B']$ on commuting pairs $(A',B')$, we may set $\Phi(0,q_{j_2}\ketbrar{\beta_{j_2}})=0$, and otherwise $\Phi(0,q_{j_2}\ketbrar{\beta_{j_2}})=q_{j_2}\Phi(\ketbrar{\alpha_k},\ketbrar{\beta_k})$. We conclude that \eqref{eq:decompositiontemperate} with the conventions \eqref{eq:define0trA}, \eqref{eq:define0trB}, \eqref{eq:define0otherwise} are correct for $\K=\R_+,\R_+^{\rm op}$. 
    
    Next, assume $\K=\T\R_+$ or $\T\R_+^{\rm op}$. Recall that $\Phi$ must restrict to $\Phi(A,B)=1$ for commuting pairs $(A,B)\neq 0$ and $\Phi(0,0)=0$. Assume there is a $j_1$ such that $\ket{\alpha_{j_1}}\neq0$ and $\ket{\beta_{j_1}}=0$. Using \eqref{eq:specialblock1}, \eqref{eq:specialblock2}, we now have 
    \begin{equation}
        \Phi\big(p_k/2\ketbrar{\alpha_k}\oplus p_{j_1}\ketbrar{\alpha_{j_1}},q_k/2\ketbrar{\beta_k}\oplus0\big)=\Phi\big(\ketbrar{\alpha_k},\ketbrar{\beta_k}\big). 
    \end{equation}    
    Then, using similar arguments as above using \eqref{eq:split1}, \eqref{eq:split2}, and remembering that addition in $\T\R_+,\T\R_+^{\rm op}$ is defined by taking the maximum of two numbers, we may conclude that 
    \begin{equation}
        \Phi(A,B)=\Phi\big(p_k\ketbrar{\alpha_k}\oplus\tilde A,q_k\ketbrar{\beta_k}\oplus\tilde B\big). 
    \end{equation}
    The case of $j_2$ such that $\ket{\alpha_{j_2}}=0$ and $\ket{\beta_{j_2}}\neq0$ can be handled analogously. We conclude that \eqref{eq:decompositiontropical} with the convention \eqref{eq:define0otherwise} is correct for $\K=\T\R_+,\T\R_+^{\rm op}$. 
    
    In the next part of the proof, we will first exclude certain ranges that the $\alpha$ and $z$ parameters may be in by providing counterexamples to the monotonicity property of homomorphisms. Similarly, we will use monotonicity to show that there are no non-zero monotone derivations. Finally, we prove that the homomorphisms that remain are in fact monotone. 
    
    Fix a monotone homomorphism $\Phi:S_{\rm m.r.}\rightarrow\K$, where $\K\in\{\R_+,\R_+^{\rm op},\allowbreak\T\R_+,\T\R_+^{\rm op}\}$. From the results in Propositions \ref{propo:Uhlmann} and \ref{prop:z} on majorization of pairs of pure states, and using the monotonicity of $\Phi$, it follows that the parameter $z$ must be non-positive when ${\mb K}=\R_+,\T\R_+$ and non-negative when ${\mb K}=\R_+^{\rm op},\T\R_+^{\rm op}$. 
    
    Let us assume $\mb K=\R_+^{\rm op}$ first. In the following, we will consider case-by-case the different ways that $\Phi$ restricts to commuting pairs. 
    
    The quantum channel arising from measuring a qubit in a fixed basis $\ket{0},\ket{1}$ yields the following majorization relation: 
    \begin{equation}
    \label{eq:alphazrelationship1}
        \left(\,\ketbrar{\psi(\pi/4)},\ketbrar{0}\,\right)\succeq(1/2,1)\boxplus(1/2,0), 
    \end{equation}
    where $|\psi(\theta)\rangle$ is defined as in \eqref{eq:anglevectors}. Consider the case that $\Phi$ is of the form \eqref{eq:Phitemperate} with $\alpha\in(0,1)$ when restricted to commuting pairs. Evaluating $\Phi$ on both sides of the inequality \eqref{eq:alphazrelationship1}, it follows by monotonicity that 
    \begin{equation}
    \label{eq:alphazrelationship2}
        \frac{1}{2^z}\leq\frac{1}{2^\alpha}. 
    \end{equation}
    By switching around the states in the two pairs in \eqref{eq:alphazrelationship1}, we have similarly
    \begin{equation}
    \label{eq:alphazrelationship3}
        \frac{1}{2^z}\leq\frac{1}{2^{1-\alpha}}. 
    \end{equation}    
    Hence \sloppy{${z\geq\max\{\alpha,1-\alpha\}}$}. It follows by the decomposition \eqref{eq:decompositiontemperate} with the convention \eqref{eq:define0otherwise} that $\Phi$ must be as in \eqref{eq:temperate} with $\sloppy{\alpha\in(0,1)}$ and ${z\geq\max\{\alpha,1-\alpha\}}$. 
    
    Next, consider the case that $\Phi$ is one of \eqref{eq:Phicharacter} when restricted to a commuting pair. When it restricts to the second form in \eqref{eq:Phicharacter}, monotonicity and \eqref{eq:alphazrelationship1} imply that 
    \begin{equation}
        \frac{1}{2^z}\leq\frac{1}{2}. 
    \end{equation}
    Hence, $z\geq1$. Similarly, when $\Phi$ restricts to the first form in \eqref{eq:Phicharacter}, by switching around the states in the two pairs in \eqref{eq:alphazrelationship1}, we find $z\geq1$. It follows for these latter two cases that $\Phi$ is equal to \eqref{eq:temperate} with $\alpha=0$ or $1$, and $z\geq1=\max\{\alpha,1-\alpha\}$. 

    A measure-and-prepare channel establishes the following majorization relation: 
    \begin{equation}\label{eq:preparation}        (1/2,1/2)\boxplus(1/2,0)\boxplus(0,1/2)\succeq(1/2,1/2)\boxplus(1/2\ketbrar{0},1/2\ketbrar{\psi(\pi/4)}). 
    \end{equation}
    Consider the remaining case that $\Phi$ is equal to $\tr[A]$ or $\tr[B]$ on commuting pairs. By monotonicity, it follows from \eqref{eq:preparation} that 
    \begin{equation}\label{eq:preparationinequality}
        1\leq\frac{1}{2}+\frac{1}{2}\frac{1}{2^z}.
    \end{equation} 
    Since $z\geq0$, this implies that $z=0$. Hence, using the convention \eqref{eq:define0trA} or \eqref{eq:define0trB}, $\Phi$ is one of the degenerate homomorphisms $\Phi(A,B)=\tr[A]$ or $\tr[B]$ for all $(A,B)\in S_{\rm m.r.}$. 

    Assume now that ${\mb K}=\R_+$. Then $\Phi$ equals $\tr[A]$ or $\tr[B]$ on commuting pairs. Monotonicity and \eqref{eq:preparation} now yield \eqref{eq:preparationinequality} with the inequality reversed. Since $z\leq0$, it follows that $z=0$. Again, $\Phi$ is one of the degenerate cases $\Phi(A,B)=\tr[A]$ or $\tr[B]$ for all $(A,B)\in S_{\rm m.r.}$. 

    Next, assume ${\mb K}=\T\R_+^{\rm op}$. Then $\Phi$ is the highly degenerate $\Phi(A,B)=1$ when $(A,B)\neq 0$ and $\Phi(0,0)=0$ on commuting pairs. Also, $z\geq0$. Using the decomposition \eqref{eq:decompositiontropical} with the convention \eqref{eq:define0otherwise}, we see that when $z=0$ we get the degenerate homomorphism that is equal to $1$ for all non-zero elements in $S_{\rm m.r.}$. When $z>0$, we find that $\Phi$ must be as in \eqref{eq:tropicalz}. 
    
    Finally, assume ${\mb K}=\T\R_+$. Then $\Phi$ is again the highly degenerate $\Phi(A,B)=1$ when $(A,B)\neq 0$ and $\Phi(0,0)=0$ on commuting pairs. The majorization relation \eqref{eq:preparation} implies that 
    \begin{equation}
    \label{eq:zrelationship2}
        1\geq\max\left\{1,\frac{1}{2^z}\right\}, 
    \end{equation}    
    Since $z\leq0$, this implies $z=0$. Hence, using the decomposition \eqref{eq:decompositiontropical} with the convention \eqref{eq:define0otherwise}, we get the degenerate homomorphism that is equal to $1$ for all non-zero elements in $S_{\rm m.r.}$. 

    We now turn to studying the monotone derivations on $S_{\rm m.r.}$ at the two degenerate homomorphisms $\Phi(A,B)=\tr[A]$ or $\tr[B]$ for all $(A,B)\in S_{\rm m.r.}$. Let $\Delta:S_{\rm m.r.}\rightarrow\R$ be a monotone derivation at $\Phi(A,B)=\tr[A]$, \sloppy{${(A,B)\in S_{\rm m.r.}}$}. Observe that for any $(A,B)\in S_{\rm m.r.}$, 
    \begin{equation}
        (\tr[A],\tr[B])\boxplus(1,1)\preceq(A,B)\boxplus(1,1)\preceq(1,1)\boxplus(\tr[A],0)\boxplus(0,\tr[B]), 
    \end{equation}
    where the second majorization is achieved by using a measure-and-prepare channel. By interchangeability, we may assume that $\Delta$ vanishes on commuting pairs. Since the outer left and outer right sides of the above inequalities are commuting pairs, monotonicity and additivity of $\Delta$ yields 
    \begin{equation}
        0=\Delta((A,B)\boxplus(1,1))=\Delta(A,B)+\underbrace{\Delta(1,1)}_{=0}. 
    \end{equation}
    Hence, $\Delta$ vanishes everywhere on $S_{\rm m.r.}$. The same conclusion holds for any derivation at the other degenerate monotone homomorphism $\Phi(A,B)=\tr[B]$, $(A,B)\in S_{\rm m.r.}$. 

    Finally, we show that the maps in \eqref{eq:temperate} and \eqref{eq:tropicalz} are in fact monotone homomorphisms with trivial kernel. Additivity and multiplicativity are immediate from the decompositions \eqref{eq:decompositiontemperate}, \eqref{eq:decompositiontropical} together with the multiplicativity of the factors appearing in the decompositions. Also, they all have a trivial kernel, since for every non-zero $(A,B)\in S_{\rm m.r.}$ there is at least one $k$ such that $\braket{\alpha_k}{\beta_k}\neq0$. The fact that they are all monotone can be seen as follows. The homomorphisms $\Phi_{\alpha,z}$, $\Phi^\T$ give rise to the relative entropies $\hat{D}_{\alpha,z}$, $\hat{D}^\T$ defined in \eqref{eq:alphazexplicit}, \eqref{eq:DT}. Since DPI is known to be satisfied for the $\alpha$-$z$ relative entropies when $\alpha\in(0,1)$ and $z\geq\max\{\alpha,1-\alpha\}$ \cite{Zhang2020} (and hence also for our version with a different prefactor), the pointwise limits 
    \begin{equation}
        \hat D^\T=\lim_{z\rightarrow\infty}\hat{D}_{\alpha,z}\ \text{(independent of}\ \alpha\text{)}, 
    \end{equation}
    and, for $z\geq1$, 
    \begin{equation}
        \hat D_{0,z}=\lim_{\alpha\rightarrow0}\hat{D}_{\alpha,z},\quad\hat D_{1,z}=\lim_{\alpha\rightarrow1}\hat{D}_{\alpha,z} 
    \end{equation}    
    also satisfy DPI. Since these relative entropies satisfy DPI and the prefactors in front of $\log$ are negative, the homomorphisms that they are associated with (namely $\Phi_{\alpha,z}$ for $\alpha\in[0,1]$, $z\geq\max\{\alpha,1-\alpha\}$, and $\Phi^\T$) are monotone according to the opposite ordering of $\R_+^{\rm op}$, $\T\R_+^{\rm op}$. Since 
    \begin{equation}
        \Phi_z^\T(A,B)=\left(\Phi^\T(A,B)\right)^z
    \end{equation}
    for all $(A,B)\in S_{\rm m.r.}$ and $z>0$, we find that $\Phi_z^\T$ for any $z>0$ is also monotone. 
    \end{proof}

    \begin{remark}\label{rmk:Seo}
    For the semiring $S_{\rm e.o.}$ we do not have a complete picture: we know that all monotone homomorphisms in Proposition \ref{prop:monhom} associated with $S_{\rm m.r.}$ are also associated with $S_{\rm e.o.}$, but there might be more. Namely, it is still an open question whether the homomorphisms for $\K=\R_+$ or $\R_+^{\rm op}$ that have the same form as \eqref{eq:temperate}, but with parameters $\alpha=0$, $z<1$, or $\alpha=1$, $z<1$, are monotone or not. The same question of monotonicity can also be asked for the homomorphisms for $\K=\T\R_+$ that are similar to \eqref{eq:tropicalz}, but with $z<0$. Additionally, there might exist monotone derivations that do not vanish everywhere. In the following, we will show that we can rule out the existence of any $\Phi$ other than the aforementioned additional ones or the ones associated with $S_{\rm m.r.}$, using techniques similar to the proof of Proposition \ref{prop:monhom}. We suspect that these additional homomorphisms are in fact not monotone either, since they correspond to parameters $\alpha,z$ for which the $\alpha$-$z$ relative entropy is known not to satisfy DPI for general pairs of quantum states (with some conditions on the support). However, in principle, DPI might be satisfied when restricting to the specific pairs of states we consider in this work. 

    Fix a monotone homomorphism $\Phi:S_{\rm e.o.}\rightarrow\K$, where $\K\in\{\R_+,\R_+^{\rm op},\allowbreak\T\R_+,\T\R_+^{\rm op}\}$. First, let us summarize results from \cite{farooq2024} on the monotone homomorphisms and derivations associated to the \emph{matrix majorization semiring of length} $2$, which consists precisely of all commuting pairs in $S_{\rm e.o.}$. When $\K=\R_+^{\rm op}$, $\Phi$ restricts to \eqref{eq:Phitemperate} for some $\alpha\in[0,1]$, on commuting pairs. When $\K=\R_+$, $\Phi$ restricts to \eqref{eq:Phitemperate} for some $\alpha\in(-\infty,0]\cup[1,\infty)$. Note that for $\K=\R_+^{\rm op}$ or $\R_+$ the cases $\alpha=0$, respectively $\alpha=1$, correspond to the degenerate homomorphisms $\tr[B]$, respectively $\tr[A]$, for a commuting pair $(A,B)$. When $\K=\T\R_+$, $\Phi$ on commuting pairs restricts to 
    \begin{equation}\label{eq:Phitropicalcommuting}
    \Phi(A,B)=\max_{1\leq i\leq n}\left(\frac{p_i}{q_i}\right)^\beta, 
    \end{equation}
    for some $\beta\in\R$. For $\beta=0$, this corresponds to the highly degenerate $\Phi(A,B)=1$ for $(A,B)\neq 0$ a commuting pair and $\Phi(0,0)=0$. When $\K=\T\R_+^{\rm op}$, $\Phi$ always restricts to this highly degenerate case on commuting pairs. The two degenerate homomorphisms $\tr[A],\tr[B]$ for commuting pairs $(A,B)$ both have non-zero monotone derivations associated with them (similar in form to the Kullback-Leibler divergence), but we will not use these here. 

    Note that the pair on the right-hand side in \eqref{eq:alphazrelationship1} is not inside $S_{\rm e.o.}$. However, for the case $\K=\R_+^{\rm op}$ and $\alpha\in(0,1)$, we can still deduce that $\Phi$ is of the form \eqref{eq:temperate} with $z\geq\max\{\alpha,1-\alpha\}$. To see this, consider again measuring a qubit in a fixed basis $\ket{0},\ket{1}$: 
    \begin{equation}
    \label{eq:alphazrelationship1eo}
        \left(\,\ketbrar{\psi(\pi/4)},\ketbrar{\psi(\theta)}\,\right)\succeq(1/2,\cos^2\theta)\boxplus(1/2,\sin^2\theta), 
    \end{equation}
    where $\theta\in(0,\pi/2)$, which ensures that the pair on the right-hand side above is in $S_{\rm e.o.}$. Assume $\K=\R_+^{\rm op}$ and $\alpha\in(0,1)$. Using monotonicity of $\Phi$, \eqref{eq:alphazrelationship1eo} yields 
    \begin{equation}
    \label{eq:alphazrelationship2eo}
        \left(\cos(\theta-\pi/4)\right)^{2z}\leq\frac{1}{2^\alpha}\left((\cos\theta)^{2-2\alpha}+(\sin\theta)^{2-2\alpha}\right) 
    \end{equation}
    for all $\theta\in(0,\pi/2)$, and by switching around the states in the pairs in \eqref{eq:alphazrelationship1eo}, we have similarly 
    \begin{equation}
    \label{eq:alphazrelationship3eo}
        \left(\cos(\theta-\pi/4)\right)^{2z}\leq\frac{1}{2^{1-\alpha}}\left((\cos\theta)^{2\alpha}+(\sin\theta)^{2\alpha}\right). 
    \end{equation}
    Taking the limit $\theta\rightarrow\pi/2$ in the above two inequalities allows us to conclude that $\sloppy{z\geq\max\{\alpha,1-\alpha\}}$ for $\alpha\in(0,1)$. Note that for $\alpha=0$ or $1$, we are not able to use this method to conclude that $z\geq1=\max\{\alpha,1-\alpha\}$. 

    Next, assume $\K=\R_+$ and $\alpha\in(-\infty,0)\cup(1,\infty)$. Note that \eqref{eq:alphazrelationship2eo} and \eqref{eq:alphazrelationship3eo} hold with the inequality reversed. Assume $\alpha>1$. When taking the limit $\theta\rightarrow\pi/2$, the left-hand side of \eqref{eq:alphazrelationship2eo} converges to $\frac{1}{2^z}$, but the right-hand side diverges to positive infinity, leading to a contradiction. When $\alpha<0$, a similar argument using \eqref{eq:alphazrelationship3eo} leads to a contradiction. This rules out the existence of any monotone homomorphism on $S_{\rm e.o.}$ associated with $\alpha$ outside $[0,1]$. 

    Now, we assume $\K=\T\R_+$ and $\Phi$ is of the form \eqref{eq:Phitropicalcommuting} for some $\beta\neq0$ when restricted to commuting pairs. Now \eqref{eq:alphazrelationship1eo} implies 
    \begin{equation}
    \label{eq:TR+}
        \left(\cos(\theta-\pi/4)\right)^{2z}\geq\frac{1}{2^\beta}\max\left\{(\cos\theta)^{-2\beta},(\sin\theta)^{-2\beta}\right\} 
    \end{equation}
    for all $\theta\in(0,\pi/2)$. Assume $\beta>0$. Then for $\theta$ restricted to the smaller interval $[\pi/4,\pi/2)$, 
    \begin{equation}
        \left(\cos(\theta-\pi/4)\right)^{2z}\geq\frac{1}{2^\beta}(\cos\theta)^{-2\beta}. 
    \end{equation}
    When taking the limit $\theta\rightarrow\pi/2$, the left-hand side converges to $\frac{1}{2^z}$, but the right-hand side diverges to positive infinity, leading to a contradiction. When $\beta<0$, an analogous argument where we switch around the states in the pairs in \eqref{eq:alphazrelationship1eo} leads to a contradiction. Note that we cannot rule out the case $\beta=0$ using this method. Hence there might exist monotone homomorphisms that are similar to \eqref{eq:tropicalz}, but with $z<0$. 

    Finally, assume $\K=\T\R_+^{\rm op}$. Analogously as in the case of the semiring $S_{\rm m.r.}$, we can conclude that any non-degenerate $\Phi$ must be of the form \eqref{eq:tropicalz} for some $z>0$.  
    \end{remark}

    \begin{remark}
        Proposition \ref{prop:puS} tells us that any $(A,B)\in S_{\rm e.o.}$, where $A,B$ have trace $1$ and satisfy $|\braket{\alpha_i}{\beta_i}|<1$ for all $i=1,\ldots,n$, is power universal. We can now show that these are in fact all possible power universals in $S_{\rm e.o.}$, using the tropical monotone homomorphism $\Phi^\T$ defined in \eqref{eq:tropical}. 

        To see this, let $(A,B)\in S_{\rm e.o.}$ be power universal. Then there exists $m\in\N$ such that $(A,B)^{\boxtimes m}\succeq(\ketbrar{0},\ketbrar{\psi(\pi/4)})$. Hence, 
        \begin{equation}
            \left(\Phi^\T(A,B)\right)^m\leq\Phi^\T(\ketbrar{0},\ketbrar{\psi(\pi/4)})=1/2<1. 
        \end{equation}
        Therefore, $\Phi^\T(A,B)<1$, from which it follows that $|\braket{\alpha_i}{\beta_i}|<1$ for all $i=1,\ldots,n$. 
    \end{remark}


    \section{Deriving the Main Results}\label{sec:proofs}

    We next show how our main results Theorems \ref{thm:LS}, \ref{thm:approximateLS}, \ref{thm:minimal} and \ref{thm:rates} follow from the analysis of the semiring $S_{\rm m.r.}$ in the previous section. 

    \subsection{Exact Majorization}\label{sec:exactproof}

    Our result Theorem \ref{thm:LS} is a direct application of the Vergleichsstellensatz in the form of Theorem \ref{thm:Fritz2022}. Note that the monotone homomorphisms associated with $S_{\rm m.r.}$ (see Proposition \ref{prop:monhom}) give rise to the relative entropies \eqref{eq:alphazexplicit} and \eqref{eq:DT}. 

    The surjective homomorphism with trivial kernel $\|\cdot\|:S\to\R_{>0}^2\cup\{(0,0)\}$ needed in the application of Theorem \ref{thm:Fritz2022} is simply given by $\|(A,B)\|=(\tr[A],\tr[B])$. The first property in \eqref{eq:surjectivehomomorphismproperties} follows from the trace preservation property of quantum channels, and the the second property from the following majorization relations: 
    \begin{equation}
        (A,B)\succeq(\tr[A],\tr[B])=(\tr[A'],\tr[B'])\preceq(A',B')
    \end{equation}
    for any $(A,B),(A',B')\in S_{\rm m.r.}$ satisfying $\|(A,B)\|=\|(A',B')\|$. 

    Note that \eqref{eq:pu1} is automatically satisfied when condition \eqref{eq:LScondition2} is met: the latter implies that $\hat D^\T(\rho\|\sigma)>0$, which is equivalent to $\max_{i=1,\ldots,n}|\braket{\alpha_i}{\beta_i}|^{2}<1$, which in turn implies \eqref{eq:pu1}. By Proposition \ref{prop:pumr}, for $(\rho,\sigma)$ in the statement of the Theorem to be power universal it thus suffices to only require it satisfies \eqref{eq:pu2}, besides the inequalities \eqref{eq:LScondition1}, \eqref{eq:LScondition2}. 
    
    Our main result, Theorem \ref{thm:LS}, now follows immediately from Propositions \ref{prop:monhom} and \ref{prop:pumr}, and Theorem \ref{thm:Fritz2022}. 

    \subsection{Asymptotic Majorization}\label{sec:approximate}

    As a consequence of Theorem \ref{thm:LS}, we can formulate a similar result where the inequalities in terms of relative entropies are only required to be satisfied \emph{non-strictly}. In this case, large-sample and catalytic majorization only hold asymptotically, in the sense that we reach only one of the two states in the output pair exactly and the other state approximately, up to arbitrarily small error. This result is contained in Theorem \ref{thm:approximateLS}, which we prove here. 


    \begin{proof}[Proof of Theorem \ref{thm:approximateLS}]
        First, note that satisfying the inequalities in part (i) of the Theorem is equivalent to $\hat D_{\alpha,z}(\rho\|\sigma)\geq \hat D_{\alpha,z}(\rho'\|\sigma')$ for all $\alpha\in[0,1]$, $z\geq\max\{\alpha,1-\alpha\}$, and additionally $\hat D^\T(\rho\|\sigma)\geq \hat D^\T(\rho'\|\sigma')$. 

        Also, by Proposition \ref{prop:pumr}, $(\rho,\sigma)$ is power universal if and only if it satisfies \eqref{eq:pu1} and \eqref{eq:pu2}. Note that Theorem \ref{thm:Fritz2022} only requires $(\rho,\sigma)$ to be power universal in the large-sample setting. This is why we require \eqref{eq:pu2} to hold for equivalence of all conditions (i), (ii), (iii). We require property \eqref{eq:pu1} to hold both in the large-sample and catalytic setting since our construction in this proof relies on it. However, since we will invoke Theorem \ref{thm:LS}, $(\rho,\sigma)$ has to satisfy \eqref{eq:pu1} anyway if condition \eqref{eq:LScondition2} is to be met, as explained in Section \ref{sec:exactproof}. 
        
        (i) $\Rightarrow$ (ii) and (iii): We write
        \begin{equation}
        \rho'=\bigoplus_{i=1}^{n'}p'_i\ketbrar{\alpha'_i},\quad\sigma'=\bigoplus_{i=1}^{n'}q'_i\ketbrar{\beta'_i}, 
        \end{equation}
        with all $\ket{\alpha'_i},\ket{\beta'_i}$ normalized or equal to $0$. Let $J\subseteq[n']$ be the subset of all $i$ where $0<|\braket{\alpha'_i}{\beta'_i}|<1$, which is non-empty since $\rho'$, $\sigma'$ do not commute. One shows analogously as in the discussion above \eqref{eq:pairstandardform} that each pair $(\ketbrar{\alpha'_i},\ketbrar{\beta'_i})$, $i\in J$, is equivalent (according to the equivalence relation $\approx$ defined in Section \ref{sec:definingthesemirings}) to the pair 
        \begin{align}
            \big(\cos^2\theta_i&\ketbrar{e_{i,1}}+\cos\theta_i\sin\theta_i\ketbra{e_{i,1}}{e_{i,2}}\\
            &+\cos\theta_i\sin\theta_i\ketbra{e_{i,2}}{e_{i,1}}+\sin^2\theta_i\ketbrar{e_{i,2}},\ketbrar{e_{i,1}}\big), 
        \end{align}
        for some angle $\theta_i\in(0,\pi/2]$, and $\ket{e_{i,1}},\ket{e_{i,2}}$ an orthonormal basis of a $2$-dimensional Hilbert space. Choose $\eta\in(0,\pi/2]$ such that $\eta<\theta_i$ for all $i\in J$ and $\cos\eta\geq\sqrt{1-\varepsilon}$. For each $i\in J$, we define 
        \begin{equation}
            \ket{\alpha'_{\varepsilon,i}}:=\cos(\theta_i-\eta)\ket{e_{i,1}}+\sin(\theta_i-\eta)\ket{e_{i,2}}, 
        \end{equation}
        and for each $i\in[n']\setminus J$, we define $\ket{\alpha'_{\varepsilon,i}}:=\ket{\alpha'_i}$. Consider 
        \begin{equation}
            \rho'_\varepsilon:=\bigoplus_{i=1}^{n'}p'_i\ketbrar{\alpha'_{\varepsilon,i}}. 
        \end{equation}
        One computes that $|\braket{\alpha'_{\varepsilon,i}}{\alpha'_i}|=\cos(\eta)\geq\sqrt{1-\varepsilon}$ for all $i\in J$, hence the fidelity between $\rho'_\varepsilon$ and $\rho'$ satisfies 
        \begin{equation}
            F(\rho'_\varepsilon,\rho')=\left(\sum_{i=1}^{n'} p'_i|\braket{\alpha'_{\varepsilon,i}}{\alpha'_i}|\right)^2\geq1-\varepsilon. 
        \end{equation}

        Also, for all $i\in J$ 
        \begin{equation}
            |\braket{\alpha'_i}{\beta'_i}|=\cos(\theta_i)<\cos(\theta_i-\eta)=|\braket{\alpha'_{\varepsilon,i}}{\beta'_i}|, 
        \end{equation}
        hence for all $\alpha\in[0,1]$, $z\geq\max\{\alpha,1-\alpha\}$ 
        \begin{equation}
            \hat D_{\alpha,z}(\rho\|\sigma)\geq \hat D_{\alpha,z}(\rho'\|\sigma')>\hat D_{\alpha,z}(\rho'_\varepsilon\|\sigma'). 
        \end{equation}
        Assume $\hat D^\T(\rho'\|\sigma')>0$, i.e. $\max_{i=1,\ldots,n'}|\braket{\alpha'_i}{\beta'_i}|^{2}<1$. Then this maximum is obtained by some subspace in $J$, from which it follows that $\max_{i=1,\ldots,n'}|\braket{\alpha'_i}{\beta'_i}|^{2}<\max_{i=1,\ldots,n'}|\braket{\alpha'_{\varepsilon,i}}{\beta'_i}|^{2}$, whence $\hat D^\T(\rho\|\sigma)\geq \hat D^\T(\rho'\|\sigma')>\hat D^\T(\rho'_\varepsilon\|\sigma')$. In case $\hat D^\T(\rho'\|\sigma')=0$, it follows immediately that $\hat D^\T(\rho\|\sigma)>\hat D^\T(\rho'_\varepsilon\|\sigma')$, since $\hat D^\T(\rho\|\sigma)>0$ by \eqref{eq:pu1}. 

        The pairs $(\rho,\sigma)$ and $(\rho'_\varepsilon,\sigma')$ thus satisfy the strict inequalities \eqref{eq:LScondition1} and \eqref{eq:LScondition2} of Theorem \ref{thm:LS}, hence (ii) and (iii) follow. 

        (ii) or (iii) $\Rightarrow$ (i): Using the tensor additivity and the DPI property of the $D_{\alpha,z}$, large-sample majorization as in (ii) implies that for all $\varepsilon>0$ and all $\alpha\in(0,1)$, $z>\max\{\alpha,1-\alpha\}$ 
        \begin{equation}\label{eq:approximateinequality}
            D_{\alpha,z}(\rho\|\sigma)\geq D_{\alpha,z}(\rho'_\varepsilon\|\sigma') 
        \end{equation}
        for a state $\rho'_\varepsilon$ satisfying $F(\rho'_\varepsilon,\rho')\geq1-\varepsilon$. Similarly, the same is implied by catalytic majorization as in (iii). 

        We may assume that $\rho'_\varepsilon$ resides in a Hilbert space $\Hil'_\varepsilon$ possibly larger than $\Hil'$ where $\rho'$ and $\sigma'$ operate. Let $P_\varepsilon$ be the projection of $\Hil'_\varepsilon$ onto $\Hil'$ and $P_\varepsilon^\perp$ be the projection onto the orthogonal complement of $\Hil'$ within $\Hil'_\varepsilon$. Set up the channel $\mc D_\varepsilon$, $\mc D_\varepsilon(\tau)=P_\varepsilon\tau P_\varepsilon+P_\varepsilon^\perp\tau P_\varepsilon^\perp$. Denoting $P_\varepsilon\rho'_\varepsilon P_\varepsilon=:{\rho'_\varepsilon}^0$ and $P_\varepsilon^\perp\rho'_\varepsilon P_\varepsilon^\perp=:{\rho'_\varepsilon}^1$, we have $\mc D_\varepsilon(\rho'_\varepsilon)={\rho'_\varepsilon}^0\oplus{\rho'_\varepsilon}^1$ and $\mc D_\varepsilon(\sigma')=\sigma'$ when we view $\sigma'$ as a state on $\Hil'_\varepsilon$.
        Note that $\lim_{\varepsilon\rightarrow0}({\rho'_\varepsilon}^0,\sigma')=(\rho',\sigma')$. 
        
        Next, for $\alpha\in(0,1)$ and $z>\max\{\alpha,1-\alpha\}$, consider the expression 
        \begin{equation}
            \Psi_{\alpha,z}(\tau,\omega):=\tr\left[\left(\omega^\frac{1-\alpha}{2z}\tau^\frac{\alpha}{z}\omega^\frac{1-\alpha}{2z}\right)^z\right] 
        \end{equation}
        inside the $\log$ of $D_{\alpha,z}$ as in \eqref{eq:alphazdivergence}. One checks that this is additive under $\oplus$, i.e. $\Psi_{\alpha,z}(\rho_1\oplus\rho_2,\sigma_1\oplus\sigma_2)=\Psi_{\alpha,z}(\rho_1,\sigma_1)+\Psi_{\alpha,z}(\rho_2,\sigma_2)$, and is non-decreasing when applying a quantum channel by the DPI of $D_{\alpha,z}$. Then, these properties and \eqref{eq:approximateinequality} imply 
        \begin{equation}
            \Psi_{\alpha,z}(\rho,\sigma)\leq\Psi_{\alpha,z}(\rho'_\varepsilon,\sigma')\leq\Psi_{\alpha,z}({\rho'_\varepsilon}^0,\sigma')+\underbrace{\Psi_{\alpha,z}({\rho'_\varepsilon}^1,0)}_{=0}. 
        \end{equation}
        Since the $\Psi_{\alpha,z}$ are continuous on pairs of fixed dimension, and $\lim_{\varepsilon\rightarrow0}({\rho'_\varepsilon}^0,\sigma')=(\rho',\sigma')$, we conclude that the inequalities in part (i) hold. 
    \end{proof}


    \subsection{Minimality of the Family of Relative Entropies}\label{sec:minimal}

    Here we prove Theorem \ref{thm:minimal} using an idea from the proof of Corollary 2.7 in \cite{strassen1988}. 
    \begin{proof}[Proof of Theorem \ref{thm:minimal}]
        Consider the set $\bar R$, containing the parameter range $R$, defined by 
        \begin{equation}\label{eq:Rbar}
            \bar R:=\big\{\,(\alpha,z)\,|\,\alpha\in[0,1],\ z\geq\max\{\alpha,1-\alpha\}\,\big\}\cup\big\{\,(\alpha,\infty)\,|\,\alpha\in[0,1]\,\big\}, 
        \end{equation}
        and define the injective map $\varphi:\bar R\rightarrow\R^2$ by 
        \begin{equation}
            \varphi(\alpha,z)=\left\{\begin{array}{ll}\big(\alpha,\frac{z}{z+1}\big)&\text{if }z\neq\infty\\
            (\alpha,1)&\text{if }z=\infty\end{array}\right..
        \end{equation}
        We define a topology on $\bar R$ by requiring that $U$ is open in $\bar R$ if and only if $U=\varphi^{-1}(V)$ for some open $V$ in $\R^2$ with the Euclidean topology. The subspace topology that $R$ inherits from $\bar R$ coincides with the Euclidean topology. One checks that the image $\varphi(\bar R)$ is a compact subset of $\R^2$, hence $\bar R$ is compact and Hausdorff. 

        For any $(A,B)\in S_{\rm m.r.}$, the evaluation function $f_{(A,B)}:\bar R\rightarrow\R$ is defined for $(\alpha,z)\in\bar R$ when $z\neq\infty$ by 
        \begin{equation}
            f_{(A,B)}(\alpha,z):=\Phi_{\alpha,z}(A,B)=\sum_{i=1}^n(p_i)^\alpha(q_i)^{1-\alpha}|\braket{\alpha_i}{\beta_i}|^{2z}
        \end{equation}
        and when $z=\infty$ by 
        \begin{equation}
            f_{(A,B)}(\alpha,\infty):=\lim_{z\rightarrow\infty}\Phi_{\alpha,z}(A,B)=\sum_{i\in[n]\,:\,|\braket{\alpha_i}{\beta_i}|=1}(p_i)^\alpha(q_i)^{1-\alpha}, 
        \end{equation} 
        where we used that all $|\braket{\alpha_i}{\beta_i}|\leq1$ to compute the limit. Observe that $\sloppy{f_{(A,B)}\circ\varphi^{-1}:\varphi(\bar R)\rightarrow\bar R}$ is continuous since it is a composition of continuous functions on $\varphi(\bar R)\setminus([0,1]\times\{1\})$, and equals the limiting value on $[0,1]\times\{1\}$. Hence, $f_{(A,B)}$ is continuous. 

        Next, let $\mc A$ be the $\R$-algebra generated by multiplication of the functions $f_{(A,B)}$, $\sloppy{(A,B)\in S_{\rm m.r.}}$, by real numbers, and finite repetitions of point-wise addition and multiplication. Then $\mc A$ is a subalgebra of $C(\bar R,\R)$, the set of all continuous real-valued functions on $\bar R$. 
        
        Observe that $f_{(1,1)}(\alpha,z)=1$ for all $(\alpha,z)\in\bar R$, hence $\mc A$ contains a non-zero constant function. Clearly, for every two distinct $(\alpha,z),(\alpha',z')\in\bar R$ there exists a pair $\sloppy{(A,B)\in S_{\rm m.r.}}$ such that $f_{(A,B)}(\alpha,z)\neq f_{(A,B)}(\alpha',z')$, hence $\mc A$ separates points. Therefore, by the Stone-Weierstrass theorem, $\mc A$ is dense in $C(\bar R,\R)$. 

        The non-empty open subset $O$ in the statement of the Theorem contains an open ball, and without loss of generality we assume $O$ to be equal to an open ball. Clearly, there exists a continuous function $f:\bar R\rightarrow\R$ such that $f(\alpha,z)=1$ for all $(\alpha,z)$ outside $O$, and $f(\alpha_0,z_0)=-1$ for some $(\alpha_0,z_0)$ in $O$. Since $\mc A$ is dense in $C(\bar R,\R)$, $f$ can be approximated by functions in $\mc A$ with arbitrary precision, from which it follows that there exist $r_1,\ldots,r_t\in\R$ and $(A_1,B_1),\ldots,(A_t,B_t)\in S_{\rm m.r.}$ such that 
        \begin{equation}\label{eq:SWapproximation}
            r_1\Phi_{\alpha,z}(A_1,B_1)+\ldots+r_t\Phi_{\alpha,z}(A_t,B_t)\left\{\begin{array}{ll}>0&\text{if }\ (\alpha,z)\in R\setminus O\\<0&\text{if }\ (\alpha,z)=(\alpha_0,z_0)\end{array}\right..
        \end{equation}
        Here we have used the multiplicativity of the homomorphisms, which implies that $\mc A$ is the span of the set of all $f_{(A,B)}$, $(A,B)\in S_{\rm m.r.}$, with real coefficients. 
        Define the pairs 
        \begin{equation}
            (A',B'):=\left(\bigoplus_{i\,:\,r_i>0}r_iA_i,\bigoplus_{i\,:\,r_i>0}r_iB_i\right),\ (A,B):=\left(\bigoplus_{i\,:\,r_i<0}-r_iA_i,\bigoplus_{i\,:\,r_i<0}-r_iB_i\right). 
        \end{equation}
        Then, by rearranging \eqref{eq:SWapproximation}, if $(\alpha,z)\in R\setminus O$, 
        \begin{equation}\label{eq:counterexample1}
            \Phi_{\alpha,z}(A',B')>\Phi_{\alpha,z}(A,B),  
        \end{equation}
        and 
        \begin{equation}\label{eq:counterexample2}
            \Phi_{\alpha_0,z_0}(A',B')<\Phi_{\alpha_0,z_0}(A,B). 
        \end{equation}
        Note that it follows in particular that both pairs $(A,B),(A',B')$ are non-zero. 
        
        If $(A,B)$ is not non-parallel, i.e. there are $i$ such that $|\braket{\alpha_i}{\beta_i}|=1$, then for such $i$ we may rotate $\ket{\alpha_i},\ket{\beta_i}$ to ensure that $|\braket{\alpha_i}{\beta_i}|<1$. By doing so, $\Phi_{\alpha,z}(A,B)$ may decrease, hence \eqref{eq:counterexample1} remains true. By ensuring that the rotations are sufficiently small, \eqref{eq:counterexample2} also remains true. We conclude that we may assume that $(A,B)$ is non-parallel. 

        If $(A',B')$ is a commuting pair, i.e. $|\braket{\alpha_i}{\beta_i}|=1$ for all $i$, then $\Phi_{\alpha_0,z}(A',B')$ does not depend on $z$. However, $\Phi_{\alpha_0,z}(A,B)$ is strictly decreasing in $z$. We may assume that $O$ is a sufficiently small open ball such that there exists $z_1\in\R$ with $(\alpha_0,z_1)\in R\setminus O$ and $z_1<z_0$. There also exists $z_2\in\R$ with $(\alpha_0,z_2)\in R\setminus O$ and $z_0<z_2$. Then \eqref{eq:counterexample1} implies that both $\Phi_{\alpha_0,z_1}(A,B)$ and $\Phi_{\alpha_0,z_2}(A,B)$ are strictly less than the constant $\Phi_{\alpha_0,z_0}(A',B')$. Since $\Phi_{\alpha_0,z_0}(A,B)>\Phi_{\alpha_0,z_0}(A',B')$, this contradicts the fact that $\Phi_{\alpha_0,z}(A,B)$ is strictly decreasing in $z$. In conclusion, $(A',B')$ is not a commuting pair. 
        
        We add an orthogonal component to both pairs as follows: 
        \begin{equation}
            (\tilde{A},\tilde{B}):=(A,B)\boxplus(a\ketbrar{0},b\ketbrar{1}),\quad(\tilde{A'},\tilde{B'}):=(A',B')\boxplus(a'\ketbrar{0},b'\ketbrar{1}), 
        \end{equation}
        with $a,b>0$, $a',b'\geq0$ chosen such that 
        \begin{equation}
            T:=\tr[\tilde{A}]=\tr[\tilde{B}]=\tr[\tilde{A'}]=\tr[\tilde{B'}]. 
        \end{equation}
        The pairs $(\rho,\sigma):=(\tilde{A},\tilde{B})/T$ and $(\rho',\sigma'):=(\tilde{A'},\tilde{B'})/T$ are in $\mc F$, and the former is non-parallel and the later is not commutative. Since $a,b>0$, $(\rho,\sigma)$ additionally satisfies \eqref{eq:pu2}. Note that 
        \begin{align}
            &\Phi_{\alpha,z}(\rho,\sigma)=T^{-1}\Phi_{\alpha,z}(\tilde{A},\tilde{B})=T^{-1}\Phi_{\alpha,z}(A,B),\\
            &\Phi_{\alpha,z}(\rho',\sigma')=T^{-1}\Phi_{\alpha,z}(\tilde{A'},\tilde{B'})=T^{-1}\Phi_{\alpha,z}(A',B')
        \end{align}    
        for all $(\alpha,z)\in R$. Hence, it follows from \eqref{eq:counterexample1}, \eqref{eq:counterexample2} that 
        \begin{equation}\label{eq:contradiction}
            \Phi_{\alpha,z}(\rho',\sigma')>\Phi_{\alpha,z}(\rho,\sigma) 
        \end{equation}
        for all $(\alpha,z)\in R\setminus O$ and the reverse inequality for $(\alpha,z)=(\alpha_0,z_0)$.
       
        Now, \eqref{eq:contradiction} for all $(\alpha,z)\in R\setminus O$ implies the conditions \eqref{eq:smallercondition}. However, (ii) and (iii) in Theorem \ref{thm:approximateLS} are both false: if either were true, then they would imply (i), which contradicts the fact that \eqref{eq:contradiction} with the reverse inequality is true for $(\alpha,z)=(\alpha_0,z_0)$. 
    \end{proof}

    \subsection{Optimal Conversion Rate}\label{sec:rate}

    Recall that, given pairs $(\rho,\sigma)$ and $(\rho',\sigma')$, $r\geq0$ is called an achievable conversion rate if $r\leq\liminf_{n\to\infty}m_n/n$ where $(m_n)_n$ is a sequence of natural numbers such that $\big(\rho^{\otimes n},\sigma^{\otimes n}\big)\succeq\big((\rho')^{\otimes m_n},(\sigma')^{\otimes m_n}\big)$ for $n\in\N$ large enough. The optimal conversion rate $r\big((\rho,\sigma)\to(\rho',\sigma')\big)$ is defined as the supremum of all the achievable conversion rates. We now give a proof for Theorem \ref{thm:rates}, which gives an expression for the optimal conversion rate in terms of the ratio of relative entropies, analogous to the proof of Corollary 29 in \cite{verhagen2025}. 
            
    \begin{proof}[Proof of Theorem \ref{thm:rates}] 
        Define 
        \begin{equation}
            \mathcal{D}:=\big\{\,\hat{D}_{\alpha,z}\,\big|\,\alpha\in[0,1],z\geq\max\{\alpha,1-\alpha\}\,\big\}\cup\big\{\hat{D}^\T\big\}. 
        \end{equation}
        We prove that the optimal conversion rate is given by \eqref{eq:rate2} by first showing that any achievable conversion rate is upper bounded by $D(\rho\|\sigma)/D(\rho'\|\sigma')$ for any relative entropy $D$ (i.e.,\ any tensor additive monotone map on state pairs, not only for $D\in\mc D$). This proof is completely standard and does not use the earlier results. This shows that 
        \begin{equation}\label{eq:oneway}
        r((\rho,\sigma)\to (\rho',\sigma'))\leq\frac{D(\rho\|\sigma)}{D(\rho'\|\sigma')},\qquad{\rm for\ any\ relative\ entropy}\ D.
        \end{equation}
        Then we show that
        \begin{equation}\label{eq:otherway}
        r((\rho,\sigma)\to (\rho',\sigma'))\geq\min_{D\in\mc D}\frac{D(\rho\|\sigma)}{D(\rho'\|\sigma')}
        \end{equation}
        with a simple application of Theorem \ref{thm:LS}.
        
        Suppose that $r\geq0$ is $((\rho,\sigma),(\rho',\sigma'))$-achievable, i.e.,\ there is a sequence $(m_n)_{n=1}^\infty\in\N^\N$ such that $r\leq\liminf_{n\to\infty}m_n/n$ and $\big(\rho^{\otimes n},\sigma^{\otimes n}\big)\succeq\big((\rho')^{\otimes m_n},(\sigma')^{\otimes m_n}\big)$ for any $n\geq n_0$ where $n_0\in\N$ is some sufficiently large number. We fix, for now, a general relative entropy $D:S_{\rm m.r.}\setminus\{0\}\to\R$, i.e,\ 
        \begin{equation}
            D(\rho\otimes\rho'\|\sigma\otimes\sigma')=D(\rho\|\sigma)+D(\rho'\|\sigma') 
        \end{equation}
        for all $(\rho,\sigma),(\rho',\sigma')\in S_{\rm m.r.}$ and $(\rho,\sigma)\succeq(\rho',\sigma')$ $\Rightarrow$ $D(\rho\|\sigma)\geq D(\rho'\|\sigma')$. From the additivity and monotonicity it immediately follows that $m_n/n\leq D(\rho\|\sigma)/D(\rho'\|\sigma')$ for all $n\geq n_0$. From the definition of $((\rho,\sigma),(\rho',\sigma'))$-achievability it follows that, for all $\varepsilon>0$, there is $n_\varepsilon\in\N$ which we can choose to be larger than $n_0$ such that $r-\varepsilon\leq \inf_{n\geq n_\varepsilon} m_n/n$. But as $n_\varepsilon\geq n_0$, this means that $r-\varepsilon\leq D(\rho\|\sigma)/D(\rho'\|\sigma')$. As this holds for all $\varepsilon>0$, we have $r\leq D(\rho\|\sigma)/D(\rho'\|\sigma')$. Taking the supremum over achievable conversion rates and then the minimum over $D\in\mc D$, we get \eqref{eq:oneway}.
        
        Let us then assume that $r=p/q<\min_{D\in\mc D}D(\rho\|\sigma)/D(\rho'\|\sigma')$ with $p,q\in\N$. We now have, for all $D\in\mathcal{D}$, 
        \begin{align}
        D(\rho^{\otimes q}\|\sigma^{\otimes q})&=qD(\rho\|\sigma)=\frac{p}{r}D(\rho\|\sigma)>pD(\rho\|\sigma)\max_{D'\in\mc D}\frac{D'(\rho'\|\sigma')}{D'(\rho\|\sigma)}\\
        &\geq pD(\rho'\|\sigma')=D(\rho'^{\otimes p}\|\sigma'^{\otimes p}).
        \end{align}
        Note that $(\rho^{\otimes q},\sigma^{\otimes q})$ satisfies \eqref{eq:pu2} since $(\rho,\sigma)$ does. Then, according to Theorem \ref{thm:LS}, we know that $(\rho,\sigma)^{\boxtimes qn}\succeq (\rho',\sigma')^{\boxtimes pn}$ for any $n\in\N$ large enough. This means that $r=p/q$ is $((\rho,\sigma),(\rho',\sigma'))$-achievable, so that $r\leq r((\rho,\sigma)\to (\rho',\sigma'))$. Since this holds for any rational $r<\min_{D\in\mc D}D(\rho\|\sigma)/D(\rho'\|\sigma')$, we finally have \eqref{eq:otherway}. Combining our observations, we arrive at \eqref{eq:rate2}.
    \end{proof}

    \section{Discussion}

    Recall that Proposition \ref{propo:Uhlmann} gives necessary and sufficient for one-shot majorization between two pairs of pure states. There is a generalization for one-shot majorization between tuples of $d\geq3$ pure states, also due to Uhlmann \cite{Uhlmann1985}. Namely, consider two tuples of $d$ pure states $A:=\big(\ketbrar{\alpha_1},\ldots,\ketbrar{\alpha_d}\big)$ and $B:=\big(\ketbrar{\beta_1},\ldots,\ketbrar{\beta_d}\big)$. The Gramm matrix of $A$ is defined as $(G(A))_{i,j}:=\braket{\alpha_i}{\alpha_j}$, and similarly for $B$. Then, $A$ majorizes $B$ if and only if there exists a positive semi-definite matrix $M$ such that $G(A)\circ M=G(B)$, where $\circ$ is the Hadamard product (i.e. entry-wise multiplication). 

    This result leads us to attempt to generalize this paper's findings on large-sample and catalytic majorization of pairs as in \eqref{eq:pairform}, to tuples of $d\geq3$ states of this form. Namely, we can consider tuples $(\rho^{(1)},\ldots,\rho^{(d)})$, with $d\geq3$, where 
    \begin{equation}\label{eq:dtuple}
    \rho^{(k)}=\sum_{i=1}^n p_i^{(k)}\ketbrar{i}\otimes\big|\alpha_i^{(k)}\big\rangle\big\langle\alpha_i^{(k)}\big|
    \end{equation}
    are cq-states with pure components. The case of $d$-tuples of finite probability distributions was studied in \cite{farooq2024,verhagen2025}. The multivariate relative entropies involved there are generalizations of the bipartite R\'enyi relative entropies. One could try to use these results, in combination with Uhlmann's result on $d$-tuples of pure states stated above, to derive multivariate quantum relative entropies that give conditions for large-sample and catalytic majorization of $d$-tuples such as \eqref{eq:dtuple}. In this setting, the problem of identifying all monotone homomorphisms reduces to analyzing the values they take on tuples of $d$ pure states. In the $d=2$ case we have shown these to be $z$-th powers of the fidelity between the two pure states (see Proposition \ref{prop:z}). Our preliminary results suggest that in the general $d$ case these will be products of the fidelities between the states in each of the $d(d-1)/2$ pairs of pure states in the tuple, raised to a power that can be different for each pair. 

    It seems to be possible to relax the condition \eqref{eq:pu2} required for large-sample majorization, both in the exact and asymptotic setting. This requires a deeper study of the semiring $S_{\rm e.o.}$, which consists precisely of those pairs in $S_{\rm m.r.}$ that do not satisfy \eqref{eq:pu2}. In Section \ref{sec:spectrum} we found that the monotone homomorphisms associated to $S_{\rm e.o.}$ are those associated to $S_{\rm m.r.}$ as well as possibly additional ones for $\alpha=0,1$ and $z<1$, their pointwise limit $z\rightarrow-\infty$, and potentially also associated derivations. Currently we are unsure whether these extra homomorphisms are monotone or not. However, the $\alpha$-$z$ relative entropies for $\alpha=0,1$, $z<1$ have been shown to violate DPI for general pairs of states. The counterexamples found in the literature do not involve pairs of states of the form \eqref{eq:pairform} that we study. Hence, the possibility remains that these homomorphisms are in fact monotone within our restricted setting. 

    \section*{Acknowledgements}
    We thank all anonymous reviewers for their useful feedback. Also, we thank P\'eter Vrana for pointing out reference \cite{strassen1988} to us. This project is supported by the National Research Foundation, Singapore through the National Quantum Office, hosted in A*STAR, under its Centre for Quantum Technologies Funding Initiative (S24Q2d0009). 
    
    \bibliographystyle{ultimate}
    \bibliography{bibliography}

@misc{fritz2022,
	title={{Abstract Vergleichsstellens\"atze for Preordered Semifields and Semirings II}}, 
	author={Tobias Fritz},
	year={2022},
	arxivid={2112.05949},
	archivePrefix={arXiv},
	primaryClass={math.AC}
}

@ARTICLE{farooq2024,
  author={Farooq, Muhammad Usman and Fritz, Tobias and Haapasalo, Erkka and Tomamichel, Marco},
  journal={IEEE Transactions on Information Theory}, 
  title={Matrix Majorization in Large Samples}, 
  year={2024},
  volume={70},
  number={5},
  pages={3118-3144},
  doi={10.1109/TIT.2024.3352088}}

@article{fritz2023,
	author = {Fritz, Tobias},
	title = {{Abstract Vergleichsstellens\"atze for Preordered Semifields and Semirings I}},
	journal = {SIAM Journal on Applied Algebra and Geometry},
	volume = {7},
	number = {2},
	pages = {505-547},
	year = {2023},
	doi = {10.1137/22M1498413},	
	URL = {https://doi.org/10.1137/22M1498413},
	eprint = {https://doi.org/10.1137/22M1498413}	
}

@article{Mu_et_al_2020,
      title={From Blackwell Dominance in Large Samples to R\'{e}nyi Divergences and Back Again}, 
      author={Mu, Xiaosheng and Pomatto, Luciano and Strack, Philipp and Tamuz, Omer},
      journal={Econometrica},
      volume={89},
      number={1},
      pages={475-506},
      doi={ https://doi.org/10.3982/ECTA17548},
      year={2020}
}

@ARTICLE{perry2022semiring,
  author={Perry, Christopher and Vrana, P\'{e}ter and Werner, Albert H.},
  journal={IEEE Transactions on Information Theory}, 
  title={The Semiring of Dichotomies and Asymptotic Relative Submajorization}, 
  year={2022},
  volume={68},
  number={1},
  pages={311-321},
  doi={10.1109/TIT.2021.3117440}}

@article{bunth2021asymptotic,
  title={Asymptotic relative submajorization of multiple-state boxes},
  author={Bunth, Gergely and Vrana, P{\'e}ter},
  journal={Letters in Mathematical Physics},
  volume={111},
  number={4},
  pages={1--23},
  year={2021},
  doi={https://doi.org/10.1007/s11005-021-01430-0},
  publisher={Springer}
}

@article{strassen1986,
  author={Strassen, V.},
  journal={27th Annual Symposium on Foundations of Computer Science (sfcs 1986)}, 
  title={The asymptotic spectrum of tensors and the exponent of matrix multiplication}, 
  year={1986},
  volume={},
  number={},
  pages={49-54},
  doi={10.1109/SFCS.1986.52}}

@article{strassen1988,
	author = {Strassen, V.},
        journal = {Journal f\"{u}r die reine und angewandte Mathematik},
	doi = {doi:10.1515/crll.1988.384.102},
	number = {384},
	pages = {102--152},
	title = {The asymptotic spectrum of tensors.},
	volume = {1988},
	year = {1988}}

@article{strassen1991,
	author = {Strassen, V.},
        journal = {Journal f\"{u}r die reine und angewandte Mathematik},
	doi = {doi:10.1515/crll.1991.413.127},
	number = {413},
	pages = {127--180},
	title = {Degeneration and complexity of bilinear maps: Some asymptotic spectra.},
	volume = {1991},
	year = {1991}}

@article{Blackwell53,
  title = {Equivalent Comparisons of Experiments},
  author = {Blackwell, D.},
  journal = {The Annals of Mathematical Statistics},
  volume = {24},
  issue = {2},
  pages = {265--272},
  numpages = {7},
  year = {1953},
  month = {6},
  publisher = {Institute of Mathematical Statistics},
  url = {https://www.jstor.org/stable/2236332},
  doi = {10.1214/aoms/1177729032}
}

@ARTICLE{bunth2023,
  author={Bunth, Gergely and Vrana, Péter},
  journal={IEEE Transactions on Information Theory}, 
  title={Equivariant Relative Submajorization}, 
  year={2023},
  volume={69},
  number={2},
  pages={1057-1073},
  doi={10.1109/TIT.2022.3214465}}

@article{Duan2005,
  title = {Multiple-copy entanglement transformation and entanglement catalysis},
  author = {Duan, Runyao and Feng, Yuan and Li, Xin and Ying, Mingsheng},
  journal = {Phys. Rev. A},
  volume = {71},
  issue = {4},
  pages = {042319},
  numpages = {10},
  year = {2005},
  month = {Apr},
  publisher = {American Physical Society},
  doi = {10.1103/PhysRevA.71.042319},
  url = {https://link.aps.org/doi/10.1103/PhysRevA.71.042319}
}

@article{Feng_et_al_2006,
  title = {Relation between catalyst-assisted transformation and multiple-copy transformation for bipartite pure states},
  author = {Feng, Yuan and Duan, Runyao and Ying, Mingsheng},
  journal = {Phys. Rev. A},
  volume = {74},
  issue = {4},
  pages = {042312},
  numpages = {6},
  year = {2006},
  month = {Oct},
  publisher = {American Physical Society},
  doi = {10.1103/PhysRevA.74.042312},
  url = {https://link.aps.org/doi/10.1103/PhysRevA.74.042312}
}

@article{strassen1987,
author = {Strassen, V.},
journal = {Journal f\"{u}r die reine und angewandte Mathematik},
pages = {406 - 443},
title = {Relative bilinear complexity and matrix multiplication.},
url = {http://eudml.org/doc/152921},
volume = {1987},
number = {375-376},
year = {1987},
}

@misc{Bunth2024,
    author = {Bunth, Gergely},
    title = {On quantum R\'{e}nyi divergences},
    school = {Budapesti M\H{u}szaki \'{e}s Gazdas\'{a}gtudom\'{a}nyi Egyetem},
    year = {2024},
    url = {http://hdl.handle.net/10890/57305}
}

@article{Vrana2022, 
    author = {Vrana,P\'eter}, 
    title = {A Generalization of Strassen's Theorem on Preordered Semirings},
    year = {2022}, 
    volume = {39}, 
    pages = {209 - 228}, 
    journal = {Order}
}

@article{Audenaert_Datta_2015,
    author = {Audenaert, Koenraad and Datta, Nilanjana},
    title = {$\alpha$-$z$-R\'{e}nyi relative entropies},
    journal = {Journal of Mathematical Physics},
    year = {2015},
    volume = {56},
    pages = {022202},
    DOI = {10.1063/1.4906367}
}

@article{verhagen2025,
      author={Verhagen, Frits and Tomamichel, Marco and Haapasalo, Erkka},
      journal={IEEE Transactions on Information Theory}, 
      title={Matrix Majorization in Large Samples With Varying Support Restrictions}, 
      year={2025},
      volume={71},
      number={9},
      pages={6517-6545},
      doi={10.1109/TIT.2025.3585062}
}

@article{Carlen_et_al_2018,
doi = {10.1088/1751-8121/aae8a3},
url = {https://dx.doi.org/10.1088/1751-8121/aae8a3},
year = {2018},
month = {nov},
publisher = {IOP Publishing},
volume = {51},
number = {48},
pages = {483001},
author = {Eric A Carlen and Rupert L Frank and Elliott H Lieb},
title = {Inequalities for quantum divergences and the Audenaert-Datta conjecture},
journal = {Journal of Physics A: Mathematical and Theoretical},
}

@article{Uhlmann1985,
url = {https://www.physik.uni-leipzig.de/~uhlmann/PDF/Uh85b.pdf},
year = {1985},
volume = {34},
number = {6},
pages = {580-582},
author = {Armin Uhlmann},
title = {Eine Bemerkung \"uber vollst\"andig positive Abbildungen von Dichteoperatoren},
journal = {Wiss. Z. Karl-Marx-Univ. Leipzig, Math.-Naturwiss. R.},
}

@article{Jordan1875,
     author = {Jordan, Camille},
     title = {Essai sur la g\'eom\'etrie \`a $n$ dimensions},
     journal = {Bulletin de la Soci\'et\'e Math\'ematique de France},
     pages = {103--174},
     publisher = {Soci\'et\'e math\'ematique de France},
     volume = {3},
     year = {1875},
     doi = {10.24033/bsmf.90},
     language = {fr},
     url = {https://www.numdam.org/articles/10.24033/bsmf.90/}
}

@book{Aczel1989, place={Cambridge}, series={Encyclopedia of Mathematics and its Applications}, title={Functional Equations in Several Variables}, publisher={Cambridge University Press}, address={Cambridge}, author={Aczel, J. and Dhombres, J.}, year={1989}, collection={Encyclopedia of Mathematics and its Applications}}

@article{Zhang2020,
title = {From Wigner-Yanase-Dyson conjecture to Carlen-Frank-Lieb conjecture},
journal = {Advances in Mathematics},
volume = {365},
pages = {107053},
year = {2020},
issn = {0001-8708},
doi = {https://doi.org/10.1016/j.aim.2020.107053},
url = {https://www.sciencedirect.com/science/article/pii/S0001870820300785},
author = {Haonan Zhang}
}

@misc{Hiai2013,
      title={Concavity of certain matrix trace and norm functions}, 
      author={Fumio Hiai},
      year={2013},
      eprint={1210.7524},
      archivePrefix={arXiv},
      primaryClass={math.FA},
      url={https://arxiv.org/abs/1210.7524}, 
}

@incollection{Jaksic2012,
    author = {Jak\v{s}i\'{c}, V. and Ogata, Y. and Pautrat, Y. and Pillet, C.-A.},
    isbn = {9780199652495},
    title = {Entropic fluctuations in quantum statistical mechanics---an introduction},
    booktitle = {Quantum Theory from Small to Large Scales: Lecture Notes of the Les Houches Summer School: Volume 95, August 2010},
    publisher = {Oxford University Press},
    address = {Oxford},
    year = {2012},
    month = {05},
    doi = {10.1093/acprof:oso/9780199652495.003.0004},
}

@misc{Hiai_Jencova_2024,
      title={$\alpha$-$z$-R\'enyi divergences in von Neumann algebras: data-processing inequality, reversibility, and monotonicity properties in $\alpha,z$}, 
      author={Fumio Hiai and Anna Jen\v{c}ov\'{a}},
      year={2024},
      eprint={2404.07617},
      archivePrefix={arXiv},
      primaryClass={quant-ph},
      url={https://arxiv.org/abs/2404.07617}, 
}

@article{Petz_85,
    author = {Petz, D\'{e}nes},
    title = {Quasi-entropies for States of a von Neumann Algebra},
    journal = {Publ.\ Res.\ Inst,\ Math.\ Sci.},
    year = {1985},
    volume = {21},
    number = {2},
    pages = {787-800},
    DOI = {10.2977/PRIMS/1195178929}
}

@article{Petz_86,
title = {Quasi-entropies for finite quantum systems},
journal = {Reports on Mathematical Physics},
volume = {23},
number = {1},
pages = {57-65},
year = {1986},
issn = {0034-4877},
doi = {https://doi.org/10.1016/0034-4877(86)90067-4},
url = {https://www.sciencedirect.com/science/article/pii/0034487786900674},
author = {Petz, D\'{e}nes}
}

@article{Muller-Lennert_et_al_2013,
    author = {M\"{u}ller-Lennert, Martin and Dupuis, Fr\'{e}d\'{e}ric and Szehr, Oleg and Fehr, Serge and Tomamichel, Marco},
    title = {On quantum R\'{e}nyi entropies: A new generalization and some properties},
    journal = {Journal of Mathematical Physics},
    year = {2013},
    volume = {54},
    pages = {122203},
    DOI = {10.1063/1.4838856}
}

@article{Wilde2013StrongCF,
  title={Strong Converse for the Classical Capacity of Entanglement-Breaking and Hadamard Channels via a Sandwiched R{\'e}nyi Relative Entropy},
  author={Mark M. Wilde and Andreas J. Winter and Dong Yang},
  journal={Communications in Mathematical Physics},
  year={2013},
  volume={331},
  pages={593 - 622},
  url={https://api.semanticscholar.org/CorpusID:14027923}
}

@article{Bottcher2010,
title = {A gentle guide to the basics of two projections theory},
journal = {Linear Algebra and its Applications},
volume = {432},
number = {6},
pages = {1412-1459},
year = {2010},
issn = {0024-3795},
doi = {https://doi.org/10.1016/j.laa.2009.11.002},
url = {https://www.sciencedirect.com/science/article/pii/S0024379509005722},
author = {A. Böttcher and I.M. Spitkovsky}
}

@article{Alberti1980,
title = {A problem relating to positive linear maps on matrix algebras},
journal = {Reports on Mathematical Physics},
volume = {18},
number = {2},
pages = {163-176},
year = {1980},
issn = {0034-4877},
doi = {https://doi.org/10.1016/0034-4877(80)90083-X},
url = {https://www.sciencedirect.com/science/article/pii/003448778090083X},
author = {P.M. Alberti and A. Uhlmann}
}

@misc{matsumoto2014,
      title={An example of a quantum statistical model which cannot be mapped to a less informative one by any trace preserving positive map}, 
      author={Keiji Matsumoto},
      year={2014},
      eprint={1409.5658},
      archivePrefix={arXiv},
      primaryClass={quant-ph},
      url={https://arxiv.org/abs/1409.5658}, 
}

@article{Buscemi2017,
  title = {Quantum relative Lorenz curves},
  author = {Buscemi, Francesco and Gour, Gilad},
  journal = {Phys. Rev. A},
  volume = {95},
  issue = {1},
  pages = {012110},
  numpages = {12},
  year = {2017},
  month = {Jan},
  publisher = {American Physical Society},
  doi = {10.1103/PhysRevA.95.012110},
  url = {https://link.aps.org/doi/10.1103/PhysRevA.95.012110}
}
	
\end{document}